\documentclass[11pt]{article}

\newif\ifblind
\blindfalse

\usepackage{algorithm}
\usepackage{algpseudocode}
\usepackage{graphics,graphicx,fullpage,natbib, multirow}
\usepackage{amsmath,amsfonts,amssymb,verbatim,epsfig}
\usepackage[colorlinks=true,linkcolor=black!50!blue,citecolor=black!20!blue,urlcolor=black!10!blue]{hyperref}
\usepackage[dvipsnames,usenames]{color}
\usepackage{graphicx}
\usepackage{subcaption}
\usepackage[normalem]{ulem}
\usepackage{enumitem}
\usepackage{soul}
\usepackage{booktabs}
\usepackage{setspace} 
\usepackage{appendix}
\usepackage{xcolor}
\usepackage{bm}

\usepackage{titlesec}

\newtheorem{lemma}{Lemma}
\newtheorem{assumption}{Assumption}

\newtheorem{proposition}{Proposition}

\newtheorem{defin}{Definition}

\newenvironment{proof}{\noindent{\bf Proof}}{$\diamond$}

\newcommand{\bI}{ {\bf I} }
\newcommand{\bX}{ {\bf X} }
\newcommand{\bY}{ {\bf Y} }

\newcommand{\bU}{ {\bf U} }

\newcommand{\bM}{ {\bf M} }
\newcommand{\bV}{ {\bf V} }
\newcommand{\bP}{ {\bf P} }

\newcommand{\bA}{ {\bf A} }

\newcommand{\bz}{ {\bf z} }

\newcommand{\by}{ {\bf y} }
\newcommand{\bh}{ {\bf h} }

\newcommand{\bD}{ {\bf D} }
\newcommand{\bE}{ {\bf E} }
\newcommand{\bS}{ {\bf S} }

\DeclareMathOperator*{\argmin}{arg\,min}

\newcommand{\bzero}{ {\bf 0} }
\newcommand{\ba}{ {\bf a} }

\newcommand{\be}{ {\bf e} }

\newcommand{\bB}{ {\bf B} }

\newcommand{\bu}{ {\bf u} }
\newcommand{\bv}{ {\bf v} }
\newcommand{\bw}{ {\bf w} }

\newcommand{\bbeta}{\text{\boldmath $ \beta $}}

\newcommand{\XV}{\bX_U(\widetilde{\bV})}
\newcommand{\XVstar}{\bX_U(\bV'_*)}
\newcommand{\XA}{\widehat{\bX}_{\mathcal{A}^*_U}}
\newcommand{\XAi}{\widehat{\bX}_{\mathcal{A}^*_i}}
\newcommand{\XAc}{\widehat{\bX}_{\mathcal{A}_U^{*,\comp}}}
\newcommand{\XAstar}{\bX^*_{\mathcal{A}^*_U}}
\newcommand{\XAistar}{\bX^*_{\mathcal{A}^*_i}}
\newcommand{\XAcstar}{\bX^*_{\mathcal{A}_U^{*,\comp}}}
\newcommand{\hbX}{\widehat{\bX}}
\newcommand{\tV}{\widetilde{\bV}}

\newcommand{\tp}{\top}
\newcommand{\vecz}{\operatorname{vec}}
\newcommand{\norm}[1]{\| #1 \|}

\newcommand{\comp}{\mathsf{c}}

\newcommand{\vecop}[1]{\operatorname{vec}\left(#1\right)}

\title{\textbf{Group-regularized matrix factorization for fast and reliable module discovery in pan-omics pan-cancer studies}}

\date{}
\ifblind
\author{\empty}
\else
{\singlespacing
\author{Jun Young Park$^{1,2}$\thanks{Co-corresponding authors: \url{junjy.park@utoronto.ca}, \url{pwmacdonald@uwaterloo.ca}.}~ and Peter W. MacDonald$^{3*}$}
    \date{
    {\small $^1$ {Department of Statistical Sciences, University of Toronto, Toronto, ON, Canada}} \\
    {\small $^2$ {Department of Psychology, University of Toronto, Toronto, ON, Canada}} \\
    {\small $^3$ {Department of Statistics and Actuarial Science, University
    of Waterloo, Waterloo, ON, Canada}}
}
}
\fi

\begin{document}

\maketitle

{\singlespacing
\begin{abstract} \noindent 
In pan-omics pan-cancer studies, it is critical to identify latent sources of variation that are shared across particular subsets. This task often requires bidimensionally linked data matrices to be decomposed into a sum of block-sparse, low-rank modules. Existing approaches often rely on pre-specified module numbers, ranks, or post-hoc thresholding and can be sensitive to model specification when the underlying sharing structure is complex. To address these issues, we propose GL-BIDIFAC+, a group-regularized matrix factorization framework for discovering partially shared modules. It requires only an upper bound on the latent dimension and encourages module selection through group regularization with theoretically-motivated tuning parameter selection and local support recovery analysis, providing both scalability and principled guidance for module discovery. It also admits a probabilistic interpretation that enables model-based imputation of missing data. Simulation studies demonstrate accurate module recovery and favorable computational performance relative to existing approaches. We further apply GL-BIDIFAC+ to analyze the Cancer Genome Atlas data, where well-established molecular structure provides interpretable biological references. Our analysis distinguishes broad pan-cancer variation, cancer-specific subtype structure, and variation shared across cancers with related tissue origins or histologic features.

\vspace{2mm}

\noindent \textit{Keywords}: bidimensional linked matrix factorization; cancer genomics; dimension reduction; group lasso; matrix completion.
  
\end{abstract}
}

\newpage

\onehalfspacing
\section{Introduction}

Modern biomedical studies increasingly collect high-dimensional data that are shared across multiple cohorts or across multiple data platforms, which has facilitated the development of data integration methods. A representative example arises in pan-omics pan-cancer studies in the Cancer Genome Atlas (TCGA), where omics features are measured across multiple cancers and tissues \citep{weinstein2013cancer}. In such settings, an important scientific goal is to utilize these data structures to identify latent structure that is shared not only globally across all data sources, but also partially across specific subsets of cancers and/or omics. 

This has motivated a broad literature on integrative factorization for multi-omic data, where low-rank structure is used to characterize shared and source-specific variation in a principled and interpretable manner. A comprehensive review  is provided by \citet{li2025integrative}. A major line of work in this area is built on  methods related to joint and individual variation explained (JIVE), which decompose linked data into low-rank components representing common and individual structure \citep{lock2013joint}. These ideas have been extended in several directions, including the incorporation of auxiliary structure, alternative geometric formulations, and more interpretable decompositions of joint variation \citep{li2017incorporating, feng2018angle, gaynanova2019structural}. More broadly, related structured decomposition ideas have also proved useful in batch effect correction  \citep{zhang2023relief, hastings2024batch}.  However, these methods primarily consider the integration of matrices that are linked in rows (i.e., multiple cohorts, single omics) or columns (multi-omics, a single cohort), and are less well-suited to bidimensionally linked matrices.

Extending the idea of JIVE, \citet{park2020integrative} and \citet{lock2022bidimensional} proposed BIDIFAC/BIDIFAC+, which model bidimensionally linked data as a sum of low-rank \textit{modules}, where each module is supported only on a particular subset of omics types and sample groups. Notably, the analysis of the pan-omics pan-cancer TCGA data was conducted in \citet{lock2022bidimensional} with 4 omics types  and 29 cancer types. In their analysis, the maximum number of modules was fixed at 50, and the analysis identified several modules with biologically meaningful interpretations, including mRNA and methylation loadings specific to sex chromosomes, as well as score patterns that distinguished subtypes of breast cancer (BRCA) across all omics profiles and low-grade glioma (LGG) through methylation. 

The motivation for reanalyzing the TCGA data in this paper is two-fold. First, we observe suboptimal module recovery of BIDIFAC+, as we report in Section \ref{sec:sim1}, in \textit{module-sparse} scenarios, where the number of active modules is smaller than the number of possible modules. The number of candidate modules grows combinatorially with the numbers of row and column groups, making exhaustive estimation infeasible. Underspecification of the number of modules can arise in pan-omics pan-cancer data, where the number of shared and distinct patterns of biological variation is expected to be large. Second, possibly due to nuclear norm penalization based on soft thresholding of singular values used in BIDIFAC+, the estimated module ranks could become too large to be interpretable. For example, cholangiocarcinoma (CHOL) had 30 samples, but the first module reported in \citet{lock2022bidimensional}, which has support over all cancer types and omics types, already had rank higher than 30. This suggests that the modules estimated by \citet{lock2022bidimensional} could be subject to possible concerns about overlapping subspaces of modules depending on the choices of tuning parameters.

\begin{figure}[t]
    \centering
    \includegraphics[width=0.9\linewidth]{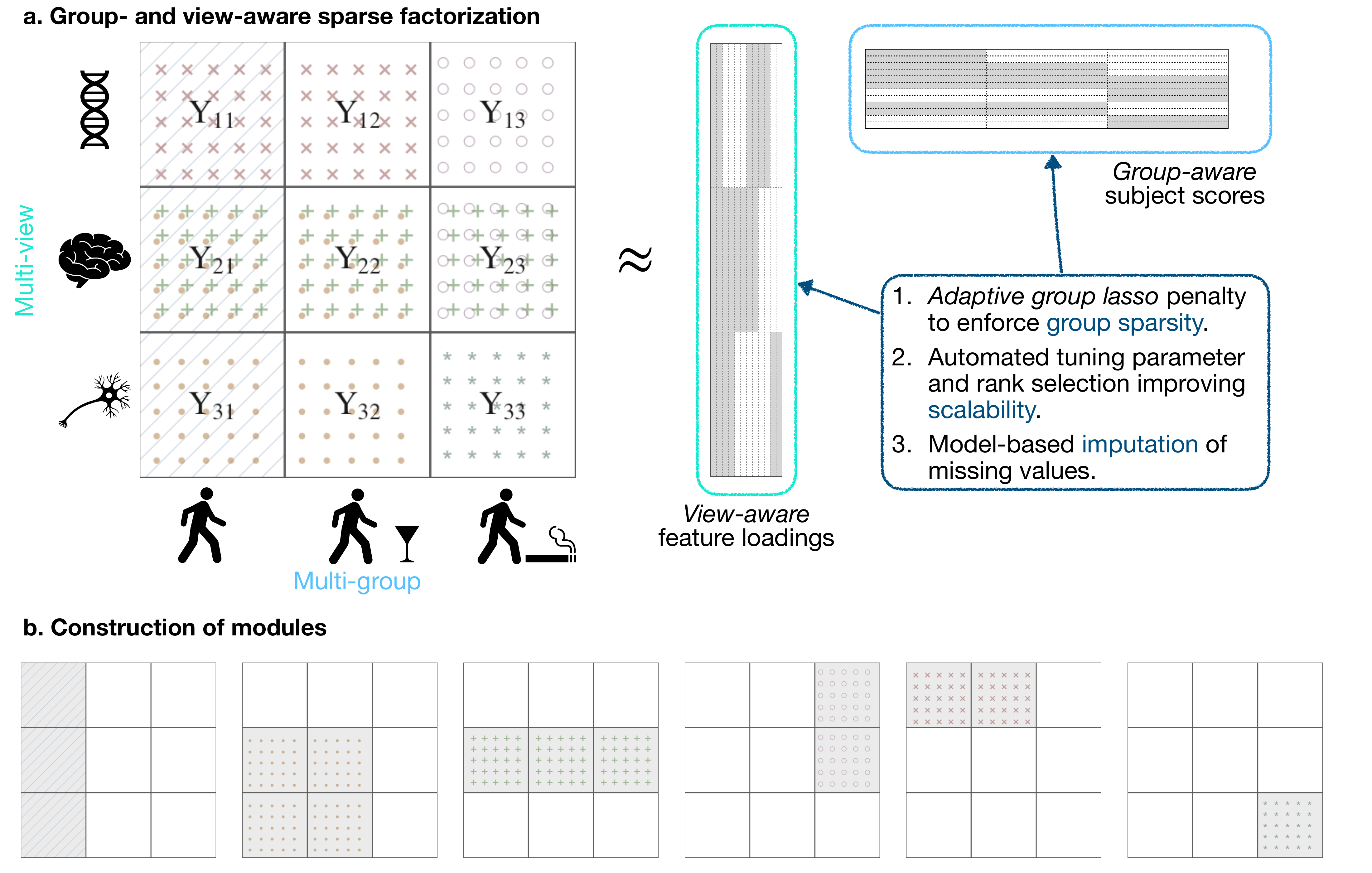}
    \caption{Schematic overview of GL-BIDIFAC+. (a) Multi-view and multi-group data are represented as a block matrix and approximated by group-aware feature loadings and subject scores. group lasso penalties imposed on the row- and column-grouped factors allow latent factors to be active only on selected views and sample groups. (b) Columns of the estimated factor matrices with the same active row- and column-group patterns are combined to form low-rank \textit{modules}, yielding biological insights on  shared or  group-specific variation.}
    \label{fig:simdesign}
\end{figure}

To address these challenges, we propose GL-BIDIFAC+ (Group Lasso-based BIDIFAC+), a group-regularized matrix factorization approach for bidimensionally linked matrices. Our strategy is to directly penalize feature loadings and subject scores with the group lasso penalties, where tuning parameters are chosen to ensure support recovery while avoiding an exhaustive grid search. A visualization of the proposed approach is provided in Figure \ref{fig:simdesign}. Motivated by the TCGA pan-omics pan-cancer analysis, we focus on the following questions:
\begin{enumerate}
    \item[Q1.] Can we discover partially shared modules reliably without enumerating the combinatorial module space? 
    \item[Q2.] Does the resulting decomposition recover known biological properties, e.g., cancer-specific subtype structure and tissue- or histology-related patterns?
    \item[Q3.] Can the analysis be scaled to pan-omics pan-cancer data while retaining the imputation capability of linked low-rank matrix factorization?
\end{enumerate}

The rest of the paper is organized as follows. The setup for decomposing bidimensionally linked matrices into a sum of modules is provided in Section \ref{sec:setup}, and several existing methods are discussed in Section \ref{sec:bidifacplus}. We present our method in Section \ref{sec:method} and investigate its theoretical properties in Section \ref{sec:theory}. Empirical validation of the proposed method is shown in Section \ref{sec:sim}, and we analyze TCGA genomic data in Section \ref{sec:data}. Discussion is provided in Section \ref{sec:discussion}.

\paragraph{Notation}
For $\ba\in\mathbb{R}^m$, we let $||\ba||_2$ be the $\ell_2$ norm of a vector $\ba$. For a matrix $\bA\in \mathbb{R}^{m\times n}$, we let $||\bA||_F$, $||\bA||_*$ be the Frobenius norm and nuclear norm of $\bA$, respectively. We define $\sigma_l(\bA)$ to be the $l$th singular value of $\bA$. For subsets $\mathcal I_1 \subseteq [m] = \{1,\ldots,m\}$ and 
$\mathcal I_2 \subseteq [n] = \{1,\ldots,n\}$, let 
$\bA_{\mathcal I_1,:}$ and $\bA_{:,\mathcal I_2}$ denote the submatrices of  $\bA$ obtained by restricting its rows to $\mathcal I_1$ and its columns to  $\mathcal I_2$, respectively. We also write $\bA_{\mathcal I_1,l}$ for the  restriction of the $l$th column of $\bA$ to rows in $\mathcal I_1$, and  $\bA_{l,\mathcal I_2}$ for the restriction of the $l$th row of $\bA$ to columns 
in $\mathcal I_2$. Lastly, $\bzero$ refers to a zero vector or matrix of appropriate size depending on the context. For matrices, $\bA\circ\bB$ refers to the Hadamard product of $\bA$ and $\bB$ of the same dimensions.

\section{Setup}\label{sec:setup}

Consider a set of $pq$ matrices $\lbrace \bY_{ij}\in\mathbb{R}^{m_i\times n_j}~:~i\in[p], j\in[q]\rbrace$ that are arranged into a block matrix $\bY\in\mathbb{R}^{m\times n}$ where $m=\sum_{i=1}^p m_i$ and $n=\sum_{j=1}^q n_j$:
\begin{align*}
    \bY=\left[\begin{array}{cccc}
      \bY_{11}   & \bY_{12} & \dots & \bY_{1q} \\
      \bY_{21}   &  \bY_{22} & \dots & \bY_{2q} \\
      \vdots & \vdots & \ddots & \vdots\\
      \bY_{p1}   & \bY_{p2} &\dots & \bY_{pq}
    \end{array}\right]\in\mathbb{R}^{m\times n}.
\end{align*}
Let $\{\mathcal{G}_1,\ldots,\mathcal{G}_{p}\}$ be a partition of $[m]$ and
$\{\mathcal{H}_1,\ldots,\mathcal{H}_{q}\}$ be a partition of $[n]$. Following  \citet{lock2022bidimensional}, our primary goal is to decompose $\bY$ as follows:
\begin{align} \label{eq:model}
    \bY=\bS+\bE, \quad \bS=\sum_{k=1}^K \bS^{(k)},
\end{align}
where $\bS^{(k)}$ is the $k$th \textit{module} and is assumed to have rank $r_k$. The entries of $\bE$ are independent, subgaussian white noise with variance at most $\sigma^2$. 
We assume that  $\bS^{(k)}$ is supported only on entries incident to the row blocks  $\mathcal{A}_k \subseteq [p]$ and the column blocks $\mathcal{B}_k \subseteq [q]$.
Each module has a factorization $\bS^{(k)} =\bU^{(k)}\bV^{(k)\top}$ (with $\bU^{(k)}\in\mathbb{R}^{m\times r_k}$ and $\bV^{(k)}\in\mathbb{R}^{n\times r_k}$), and the factors are identified up to an unknown linear transformation, as $\bU^{(k)}\bV^{(k)\top} =\bU^{(k)}\bA \bA^{-1} \bV^{(k)\top} = \bU^{(k)}\bA (\bV^{(k)}\bA^{-\top})^{\top}$
for any invertible $\bA \in \mathbb{R}^{r_k \times r_k}$.
However, since $\bA$ operates on the columns of $\bU^{(k)}$ and $\bV^{(k)}$, all factorizations of $\bS^{(k)}$ share the same pattern of row sparsity, and thus we may unambiguously write
\begin{align*}
\mathcal{A}_k = \lbrace i\in[p] :  ||\bU^{(k)}_{\mathcal{G}_i,:}||_F>0 \rbrace,\qquad
\mathcal{B}_k = \lbrace j\in[q] :  ||\bV^{(k)}_{\mathcal{H}_j,:}||_F>0 \rbrace.    
\end{align*}
We note that the decomposition in \eqref{eq:model} is not fully identifiable without additional strong conditions, e.g., mutual orthogonality of the support-restricted block matrices.
To make the module ranks and sparsity structure identifiable, we impose an additional rank additivity assumption on the decomposition.
A module is characterized by non-empty sets $\mathcal{A} \subseteq [p]$ and $\mathcal{B} \subseteq [q]$; then a module decomposition of $\bS$ is defined by the non-negative integer ranks corresponding to each subset pair. 
Further define $\bS_{\mathcal{A}\mathcal{B}} = \bS[ \cup_{i \in \mathcal{A}} \mathcal{G}_i , \cup_{j \in \mathcal{B}} \mathcal{H}_j ]$.
To help uniquely identify the ranks of a given module decomposition, we require what amounts to a linear independence and full-rank condition on each submatrix of $\bS$.

\begin{proposition} \label{prop:ident}
    Suppose $\bS$ has a module decomposition with ranks $r_k = r_*(\mathcal A_k , \mathcal B_k)$ satisfying 
    \begin{equation} \label{ident_cond}
        \operatorname{rank}(\bS_{\mathcal{A}\mathcal{B}}) = \sum_{\mathcal{A}' : \mathcal{A}' \cap \mathcal{A} \neq \emptyset} \sum_{\mathcal{B}' : \mathcal{B}' \cap \mathcal{B} \neq \emptyset} r_*(\mathcal{A}',\mathcal{B}').
    \end{equation}
    for any $\mathcal{A} \subseteq [p]$ and $\mathcal{B} \subseteq [q]$. Then $\{r_k\}$ is the unique set of module ranks which decompose $\bS$ and satisfy \eqref{ident_cond}.
\end{proposition}
If $\mathcal{A} = \{i_*\}$ and $\mathcal{B} = \{j_*\}$, then \eqref{ident_cond} says that all modules which contain $i_*$ and $j_*$ contribute their full rank to the matrix block $\bS_{\mathcal{G}_{i_*},\mathcal{H}_{j_*}}$.
If $\mathcal{A} = \{1,\ldots,p\}$ and $\mathcal{B} = \{1,\ldots,q\}$, then \eqref{ident_cond} says that the rank of $\bS$ decomposes additively over all modules; a decomposition which satisfies this condition \eqref{ident_cond} has uniquely identified module structure and ranks.

The maximum possible number of modules is bounded by $K_{\max}=(2^p-1)\times (2^q-1)$, 
which grows exponentially with respect to $p$ and $q$. For example, in  \citet{lock2022bidimensional}, $K_{\max}$ is approximately $8.05\times 10^{9}$ for $p=4$ omics types and $q=29$ cancer types. Therefore, for moderate data dimensions, it is reasonable to assume that the number of modules is much smaller than this upper bound. A critical goal is to identify $\mathcal K = \{k \in [K_{\max}] : r_k > 0\}$ accurately, while considering all $K_{\max}$ potential module structures.

Our goal here is not to compare baseline means across cancer types but to understand shared latent variation across different cancer and omics types. Therefore, following existing practices \citep{lock2013joint,park2020integrative,lock2022bidimensional}, we first center the means for each omics profile within each group of subjects (i.e., each row of $\bY_{ij}$). Then, considering that each $\bY_{ij}$ might have different noise levels, we scale each residualized $\bY_{ij}$ by its noise standard deviation, so that each $\bY_{ij}$ can be assumed to have zero mean and unit noise variance. This is achieved by noting that each $\bY_{ij}$ consists of a low-rank signal and full-rank white noise, and applying methods of \citet{gavish2017optimal} or \citet{shabalin2013reconstruction}, which provide robust estimates of the block-specific noise variances.

\section{Related methods}\label{sec:bidifacplus}

\subsection{BIDIFAC+}
 BIDIFAC+ \citep{lock2022bidimensional} considers the following convex optimization problem:
\begin{align*}
\min_{\bS^{(k)},k\in[K]} \frac{1}{2}||\bY-\sum_{k=1}^K \bS^{(k)}||_F^2+\sum_{k=1}^K\lambda_k||\bS^{(k)}||_*.
\end{align*}
The block  coordinate descent is used as updating $\bS^{(k')}$ when all the others are fixed is done in a closed-form \citep{mazumder2010spectral}. Because of the nuclear norm penalty, the resulting $\widehat{\bS}^{(k)}$ will be low-rank. Provided that $\bY$ is already processed to have unit noise variance, \citet{park2020integrative} and \citet{lock2022bidimensional} recommend $\lambda_k= \sqrt{\sum_{i\in \mathcal{A}_k} m_i} +\sqrt{\sum_{j\in \mathcal{B}_k}n_j}$, which is determined by the probabilistic upper bound on $\sigma_1(\bE)$ over the active support set of $\bS^{(k)}$. Theoretical analysis by \citet{lock2022bidimensional} showed that this specific choice of $\lambda_k$ meets necessary conditions under which the final solution $\lbrace\widehat{\bS}^{(k)}\rbrace_{k\in\widehat{\mathcal{K}}}$ does not reduce to a trivial solution obtained by soft-thresholding the singular values of each $\bY_{ij}$ separately.

Compared to the two-stage procedure that first estimates the rank of each module and then estimates the modules (e.g., \citet{lock2013joint}, \citet{gaynanova2019structural}), BIDIFAC+ allows simultaneous rank selection and estimation without the need to choose $K$ tuning parameters. However, concerns remain regarding computation and identification. First,  BIDIFAC+ is computationally infeasible without any simplification when the maximum possible number of modules is high. 
For example, the pan-omics pan-cancer example considered in \citet{lock2022bidimensional} requires $K_{\max} \approx 8.05\times 10^{9}$ separate singular value decompositions (SVDs) at each iteration. 
\citet{lock2022bidimensional} address this challenge by (i) specifying an upper bound on the number of modules and (ii) embedding a forward module-seeking step in the estimation procedure. However,  its theoretical properties are not well understood. Second, it relies on sufficiently large values of $\lambda_k$'s and, if not, each singular value of $\widehat{\bS}^{(k)}$ may not be sufficiently shrunk toward zero to induce low-rankness and rank additivity of the estimated modules. Third, the nuclear norm penalty in a single matrix completion case reduces each singular value uniformly, which might overshrink signals, a phenomenon known to appear in the nuclear norm penalization \citep{yi2023hierarchical, lock2024empirical}. 

\subsection{MOFA+}\label{sec:mofaplus}

MOFA+ \citep{argelaguet2020mofa} allows for  modeling shared and partially shared variation across multiple omics types and  sample groups through a hierarchical latent factor model. A convenient blockwise representation of MOFA+ for the Gaussian case is
\begin{align*}
\bY_{ij} = \sum_{l=1}^d \bw_{il} \bz_{jl}^\top + \bE_{ij},
\end{align*}
where \(\bw_{il} \in \mathbb{R}^{m_i}\) and \(\bz_{jl} \in \mathbb{R}^{n_j}\) denote the \(l\)th loading and factor vectors for row group \(i\) and column group \(j\), respectively. At a schematic level, MOFA+ places hierarchical priors on both
\(\bw_{il}\) and \(\bz_{jl}\). For example, one may write
$\bw_{il} \mid \alpha_{il}, s_{il} \sim \text{Spike-and-slab}(\alpha_{il}, s_{il})$ and $
\bz_{jl} \mid \tau_{jl}, r_{jl} \sim \text{Spike-and-slab}(\tau_{jl}, r_{jl})$,
where \(\alpha_{il}\) and \(\tau_{jl}\) represent shrinkage parameters based on the Automatic Relevance Determination (ARD) prior, and \(s_{il}\) and \(r_{jl}\) represent spike-and-slab-type inclusion variables.
Accordingly, ARD priors regulate factor activity across omics views through
\(\bw_{il}\) and across sample groups through \(\bz_{jl}\), while spike-and-slab priors
induce additional sparsity in both sets of latent variables. 

For estimation, MOFA+ can be written in terms of a variational evidence lower bound. At the model level, the spike-and-slab prior aims to induce sparsity, but in the actual implementation, inference is carried out via a variational approximation over Bernoulli inclusion variables and Gaussian slab components; consequently, the resulting posterior means or point estimates may remain small but nonzero rather than being exactly zero.

\section{Proposed method}\label{sec:method}

\subsection{Overview}

To mitigate the issues of BIDIFAC+ discussed in Section \ref{sec:bidifacplus}, we introduce GL-BIDIFAC+ (Group Lasso-based BIDIFAC+), which uses alternating least squares combined with group lasso penalties. Without loss of generality, let
$\bU=[\bU^{(1)},\dots,\bU^{(K)}]\in\mathbb{R}^{m\times d}$ and $\bV=[\bV^{(1)},\dots,\bV^{(K)}]\in\mathbb{R}^{n\times d}$, so that $\bS=\sum_{k=1}^K \bU^{(k)}\bV^{(k)\top}=\bU\bV^\top$.
 We view $\sum_{k=1}^K \bS^{(k)}$ as a product of $\bU$ and $\bV^\top$ for $d\leq\min(m,n)$, and consider the following penalized objective:
\begin{equation}
\argmin_{\bU\in\mathbb{R}^{m\times d},\bV\in\mathbb{R}^{n\times d}} \underbrace{ 
\frac{1}{2|\Omega|}||\bP_{\Omega}\circ (\bY-\bU\bV^\top)||_F^2}_{\mathcal{L}(\bU, \bV;\bY)}
+\lambda_U \underbrace{\sum_{l=1}^d\sum_{i=1}^{p} \omega_{il}^{(U)}||\bu_{\mathcal{G}_i,l}||_2}_{\mathcal{R}(\bU)}
+\lambda_V \underbrace{\sum_{l=1}^d\sum_{j=1}^{q} \omega_{j,l}^{(V)}||\bv_{\mathcal{H}_j,l}||_2}_{\mathcal{R}(\bV)},
\label{eq:obj_uv}
\end{equation}
where $\lambda_U,\lambda_V>0$ are tuning parameters and $d$, specified by the user, is an upper bound on the rank of the estimate of $\bS$. $\omega_{il}^{(U)}>0$ and $\omega_{jl}^{(V)}>0$ are weights for the penalties for $\bu_{\mathcal{G}_i,l}$ and $\bv_{\mathcal{H}_j,l}$, respectively, which we set to 1 in this section and later extend to adaptive weights in Section \ref{sec:agl}. $\Omega$ is the set of indices corresponding to the observed entries of $\bY$, and entries of $\bP_{\Omega}\in\mathbb{R}^{m\times n}$ are 1 on $\Omega$ and 0 elsewhere, which makes GL-BIDIFAC+ adaptable to incomplete data matrices (see Section \ref{sec:imp}). Since \eqref{eq:obj_uv} is bi-convex with respect to $\bU$ and $\bV$, $\bU$ and $\bV$ can each be updated iteratively with the weighted group lasso while fixing the other \citep{yuan2006model}. We use proximal gradient descent in each step, so that updating $\bV$ given $\bU$ has a computational cost of $O(mnd + Tnd^2)$ where $T$ is the iteration number for the proximal-gradient subproblem.

Compared to BIDIFAC+, the proposed GL-BIDIFAC+ offers improvements in computation and identification.  It mitigates these issues because GL-BIDIFAC+ requires only an upper bound $d$ on the latent dimension, rather than the exact number and ranks of modules. These are aided by the following proposition.

\begin{proposition}\label{prop1}
    Suppose that a stationary point for \eqref{eq:obj_uv} is achieved with $\widehat{\bU}$ and $\widehat{\bV}$. 
    \begin{enumerate}
        \item     If $\hat{\bu}_{\mathcal{G}_i,l}=\bzero$ for all $i\in[p]$, then $\hat{\bv}_{\mathcal{H}_j,l}=\bzero$ for all $j\in[q]$, and vice versa.

        \item $\hat{\bu}_{\mathcal{G}_i,l}\neq\bzero$ and $\hat{\bv}_{\mathcal{H}_j,l}\neq \bzero$ if and only if $\hat{\bu}_{\mathcal{G}_i,l}\hat{\bv}_{\mathcal{H}_j,l}^\top\neq \bzero$.
    \end{enumerate}
\end{proposition}
Proposition~\ref{prop1}.1 suggests that GL-BIDIFAC+ can effectively prune the rank of  $\widehat{\bS}=\widehat{\bU}\widehat{\bV}^\top$ to be lower than $d$ specified in initialization. In terms of computation, although each iteration scales with the initialized factor dimension $d$, the group lasso updates quickly induce column-wise sparsity: once a column is shrunk to exactly zero, it remains inactive and can be dropped from subsequent updates. As a result, after a few iterations the effective dimension is the number of active columns, which is often much smaller than $d$, yielding scalability. Proposition \ref{prop1}.2 is useful because, after getting $\widehat{\bU}$ and $\widehat{\bV}$, it is sufficient to check the active set of each column of $\bU$ and $\bV$ to determine the module to which each column belongs.

Recent results on overparameterized factorized matrix estimation suggest that sparsity-inducing regularization can reduce the effective number of active factors, even when the initialized factor dimension is larger than necessary \citep{li2025solving}. 
Analogously, we find that our group-regularized factorization is largely insensitive to overspecification of $d$, suggesting an analogous effective rank selection behavior. This is computationally appealing compared to other methods related to JIVE, that directly tune each rank based on resampling or cross-validation \citep{lock2013joint, gaynanova2019structural, prothero2024data}.

The next proposition shows how our  objective relates to thresholding singular values.
\begin{proposition}[Theorem 1 of \cite{fan2019factor}]\label{prop:shrinking}
For $\bS\in \mathbb{R}^{m\times n}$    with $\operatorname{rank}(\bS)=r\leq d\leq \min(m,n)$, 
\begin{align*}
\min_{\bS=\sum_{l=1}^d\bu_l\bv_l^\top} \sum_{l=1}^d ||\bu_l||_2+||\bv_l||_2=2\sum_{l=1}^r \sqrt{\sigma_l(\bS)}.
\end{align*}
\end{proposition}
When fixing $i=i^*$ and $j=j^*$ and holding all other $i,j$ fixed, the part of the objective involving the block \((i^*,j^*)\) has the same variational form as a single-matrix factorization of \(\bY_{i^*j^*}\) with a Schatten-\(1/2\)-type penalty as $\underset{\bS=\sum_{l=1}^d \bu_l \bv_l^\top}{\min}
(\lambda_1\sum_{l=1}^d ||\bu_l||_2+\lambda_2\sum_{l=1}^d ||\bv_l||_2)
=
2\sqrt{\lambda_1\lambda_2}\ \sum_{l=1}^r \sqrt{\sigma_l(\bS)}$. 
Compared to the nuclear norm penalty, the Schatten-1/2 penalty is weaker for larger singular values. This  explains how GL-BIDIFAC+ would mitigate a limitation of the structured nuclear norm penalty in BIDIFAC+ that may over-shrink when the signal-to-noise ratio is high \citep{yi2023hierarchical, lock2024empirical}.

Finally, with identifiability in mind, note that we can rewrite \eqref{eq:obj_uv} as
\begin{align*}
    &\argmin_{\bU\in\mathbb{R}^{m\times d},\bV\in\mathbb{R}^{n\times d}}
\frac{1}{2|\Omega|}|| \bP_{\Omega} \circ(\bY-\bU\bV^\top)||_F^2
+\lambda_U \mathcal{R}(\bU)
+\lambda_V \mathcal{R}(\bV) \\
= &\argmin_{\bU, \bV, \bS : \operatorname{rank}(\bS) \leq d} \left\{ \frac{1}{2|\Omega|}|| \bP_{\Omega} \circ(\bY-\bS) ||_F^2 + \min_{\bU\bV^{\tp} = \bS} \left\{ \lambda_U \mathcal{R}(\bU)
+\lambda_V \mathcal{R}(\bV) \right\} \right\}.
\end{align*}
Thus, in addition to supporting group sparsity in the entries of the optimal $(\bU,\bV)$, the penalty term of \eqref{eq:obj_uv} favors certain representatives of the class of factorized decompositions of $\bS$, and if $\lambda_U \asymp \lambda_V$, it will encourage balanced magnitudes of the two sides of the factorization.

\subsection{Tuning parameter selection}\label{sec:tuning}

We motivate tuning parameter selection by studying each iterative step of our algorithm, which is a group lasso problem (see Section~\ref{sec:theory}).

Conditional on $\bV$, updating $\bU$ in \eqref{eq:obj_uv} reduces to a group lasso problem with design matrix $\bX_U(\bV) = \bV \otimes \bI_m$. 
When the perturbation on the factor matrix $\bV$ is sufficiently small, our Proposition~\ref{prop:one_step_U}, along with well-established theory (e.g., \citet{negahban2012unified}) motivates
$$
    \lambda_{U}
\asymp \frac{\|\bV\|_{1,2}}{mn} \left( \sqrt{\max_{1 \leq i \leq p} m_i} + \sqrt{\log(pd)} \right),
$$
where $\|\cdot\|_{1,2}$ denotes the maximum column norm of a matrix.
Results for updating $\bV$ conditional on $\bU$ are analogous.
Therefore, we set the following penalty for our objective:
\begin{equation}
\lambda_U
=
2\cdot\sqrt{{\sigma}_1(\bS)}\cdot \left(
\frac{\sqrt{\max_{1 \leq i \leq p} m_i} + \sqrt{\log(pd)}}{mn}\right),
~~
\lambda_V
=
2\cdot \sqrt{{\sigma}_1(\bS)}\cdot \left(
\frac{\sqrt{\max_{1 \leq j \leq q} n_j} + \sqrt{\log(qd)}}{mn}\right).
\label{eq:lambda_used}
\end{equation}
 These choices are reasonable because, when considering the SVD of $\bS=\bU_S\bD_S\bV_S^\top$ and a balanced factorization of $\bS = \bU\bV^\top$ with $\bU = \bU_S\bD_S^{1/2}$ and $\bV = \bV_S\bD_S^{1/2}$, we have $\sigma_1(\bU)=\sigma_1(\bV) = \sigma_1(\bD_S^{1/2})=\sqrt{\sigma_1(\bS)}$. When the group sizes (e.g., $m_i,i\in[p]$) are heterogeneous, we allow $\lambda_U$ or $\lambda_V$ to differ by $i$ (or $j$) so that, for $i\in[p]$, we set $\lambda_{U,i}=2\cdot\sqrt{{\sigma}_1(\bS)}\cdot \left(
\frac{\sqrt{ m_i }+\sqrt{\log(pd)}}{mn}\right)$. 

In practice, in \eqref{eq:lambda_used}, $\sigma_1(\bS)$ is unknown and $d$ must be specified. For estimating $\sigma_1(\bS)$, we use empirical variational Bayes shrinkage results from \citet{nakajima2012perfect} and \citet{lock2024empirical}, yielding a closed-form estimator. For a threshold $\theta = \sqrt{m+n+\sqrt{mn(\kappa+1/\kappa)}}$ with $\kappa$ being the solution of $f(x)=x\sqrt{n/m}\log(x\sqrt{m/n}+1)+x\sqrt{m/n}\log(x\sqrt{n/m}+1)-1=0$,
\begin{align}\label{eq:sigma1S}
\hat{\sigma}_1(\bS) \;=\;
\begin{cases}
0 & \text{if }\sigma_1(\bY) < \theta,\\
\dfrac{\sigma_1(\bY)^2-(m+n)+\sqrt{(\sigma_1(\bY)^2-(m+n))^2-4mn}}{2\sigma_1(\bY)}
& \text{otherwise.}
\end{cases}    
\end{align}
We found this choice to work well across a range of signal-to-noise ratios. For choosing $d$, we note that $\sigma_1(\bE)$ is probabilistically bounded by $\sqrt{m}+\sqrt{n}$ \citep{rudelson2010non}; therefore, we set $d= \sum_{l=1}^{\min(m,n)} I(\sigma_l(\bY)>\sqrt{m}+\sqrt{n})$.
Our complete algorithm for GL-BIDIFAC+ is summarized in Algorithm \ref{alg:gl-bidifac}.

\begin{algorithm}[t]
\caption{GL-BIDIFAC+ with preprocessed $\bY$}
\label{alg:gl-bidifac}
\begin{algorithmic}[1]
\Statex \textbf{Initialization}
\State Set $t=1$ and $d=\sum_{l=1}^{\min(m,n)}I({\sigma}_l(\bY)>\sqrt{m}+\sqrt{n})$.
\State Initialize $\bU^{(0)}$, $\bV^{(0)}$, and $\bS^{(0)}$ following Section \ref{sec:init}, and adaptive penalty weights following Section~\ref{sec:agl}.
\State Compute $\hat{\sigma}_1(\bS)$ using \eqref{eq:sigma1S}.
\State Set $\lambda_U
=
2\cdot\sqrt{\hat{\sigma}_1(\bS)}\cdot \left(
\frac{\sqrt{\max m_i }+\sqrt{\log(pd)}}{mn}\right)$ and $
\lambda_V
=
2\cdot \sqrt{\hat{\sigma}_1(\bS)}\cdot \left(
\frac{\sqrt{\max n_j}+\sqrt{\log(qd)}}{mn}\right)$.

\vspace{0.25em}
\Statex \textbf{Estimation}
\While{$t<T$ and $||\bS^{(t)}-\bS^{(t-1)}||_F / ||\bS^{(t-1)}||_F\ \geq\epsilon$}
    \State Update $\bV^{(t)}$ given $\bU^{(t-1)}$ by solving \eqref{eq:obj_uv} with  $\lambda_V$.
    \State Update $\bU^{(t)}$ given $\bV^{(t)}$ by solving \eqref{eq:obj_uv} with   $\lambda_U$.
    \State Prune $l$th column of $\bU^{(t)}$ and $\bV^{(t)}$ for each $l$ when $\bu^{(t)}_{:,l}=\bzero$ or $\bv^{(t)}_{:,l}=\bzero$.
    \State  Set 
$\bS^{(t)}=\bU^{(t)}\bV^{(t)\top}$.
    \State $t\gets t+1$

\EndWhile

\State Set $\widehat{\bU}\gets \bU^{(t)}$, $\widehat{\bV}\gets \bV^{(t)}$, $\widehat{\bS}\gets \bS^{(t)}$.

\vspace{0.25em}
\Statex \textbf{Module separation}

\State Let $A_l = \{i\in[p] : ||\hat{\bu}_{\mathcal{G}_i,l}||_2>0\}$ and
$B_l = \{j\in[q] : ||\hat{\bv}_{\mathcal{H}_j,l}||_2>0\}$ for each column $l$ of $\widehat{\bU}$ and $\widehat{\bV}$.

\State Group columns of $\widehat{\bU},\widehat{\bV}$ by identical  $(A_l,B_l)$ to 
identify columns $C_k$ that belong to the $k$th module.

\State Compute $\widehat \bS^{(k)}=\sum_{l\in C_k}\hat{\bu}_{:,l}\hat{\bv}_{:,l}^\top
$.

\State Set $\widehat{K}$ to be the number of non-empty modules.

\vspace{0.25em}
\State \Return  $\{\widehat{\bS}^{(k)}\}_{k=1}^{\widehat K}$.
\end{algorithmic}
\end{algorithm}

\subsection{Adaptive group lasso adjustment}\label{sec:agl}

Although the penalty choice in \eqref{eq:lambda_used} is well-motivated, in practice it may still be overly aggressive when the signal-to-noise ratio is small. Weak but nonzero group-specific factors may be shrunk to zero too early in the alternating updates, which can lead to underestimation of module supports and poorer signal recovery. To reduce this bias, we consider an adaptive group lasso modification of \eqref{eq:obj_uv} \citep{wang2008note}.

Let $\widetilde{\bU}$ and $\widetilde{\bV}$ be pilot estimates obtained from GL-BIDIFAC+ with the same initialization but a smaller  penalty level. For each $i,j,l$, define the adaptive weights as follows:
\begin{align*}
{\omega}^{(U)}_{i,l}= \frac{(||\widetilde{\bu}_{G_i,l}||_2+\tau_U)^{-\gamma}}{\text{median}_{i,l}\lbrace||\widetilde{\bu}_{G_i,l}||_2+\tau_U)^{-\gamma}\rbrace} \qquad \text{and}
\qquad
{\omega}^{(V)}_{j,l}=\frac{(||\widetilde{\bv}_{H_j,l}||_2+\tau_V)^{-\gamma}}{\text{median}_{j,l}\lbrace||\widetilde{\bv}_{H_j,l}||_2+\tau_V)^{-\gamma}\rbrace},    
\end{align*}
where $\tau_U,\tau_V>0$ are small constants and $\gamma>0$ controls the strength of adaptation (set to $\gamma=2$ in this paper). 
Since $h(u) = (u + \tau)^{-\gamma}$ is a decreasing function of $u$, well chosen weights should serve to increase penalties on inactive groups and decrease penalties on active groups, widening the valid range of tuning parameters (see Proposition \ref{prop:one_step_U}).
The denominator terms center the weights around one to maintain the rates of growth for $\lambda_U$ and $\lambda_V$ given in \eqref{eq:lambda_used}.

\subsection{Initialization}\label{sec:init}

Assuming the overall rank of $\bS$ is known or can be conservatively estimated, we initialize the problem using the scaled rank-$d$ SVD of the observed $\bY$ for $d \geq r$. In GL-BIDIFAC+, however, the columns of $\bU$ and $\bV$ are expected to have group-structured sparse supports, and an arbitrary orthogonal rotation may mix module-specific directions and obscure this structure. Following Proposition \ref{prop:one_step_U}, the group lasso updates may therefore benefit from an initialization that is better aligned with the module structure, up to a signed permutation of the columns and possible within-module rotation and rescaling. Hence, we apply an additional \emph{varimax} rotation to the leading right or left singular vectors of $\bY$ to encourage sparse, module-aligned initial factors. \citet{rohe2023vintage} show that this classical varimax procedure can provably recover factors up to column permutation, rather than orthogonal transformation, in a different but related setting of principal components analysis, where the entries of the left-hand side factor matrix are random with \emph{leptokurtotic} distributions, including some sparse settings. In our experiments, this choice performed better than the (unrotated) SVD initialization in our simulations, and we leave the theoretical properties of varimax-based initialization for future work.

\subsection{Imputation}\label{sec:imp}

We note that \eqref{eq:obj_uv}  can be viewed as maximizing the log-posterior of $\{\bU,\bV\}$ under (i) a Gaussian likelihood for $\bY$ corresponding to $\mathcal{L}(\bU,\bV;\bY)$ and (ii) independent multivariate Laplace priors on each grouped factor $\bu_{\mathcal{G}_i,l}$, $i\in[p]$, $l\in[d]$, and $\bv_{\mathcal{H}_j,l}$, $j\in[q]$, $l\in[d]$, corresponding to the group lasso penalties $\mathcal{R}(\bU)$ and $\mathcal{R}(\bV)$. Alternatively, \citet{casella2010penalized} showed that the group lasso can be represented hierarchically through Gaussian priors with group-specific variance parameters and Gamma hyperpriors on those variances. This interpretation suggests a natural extension of GL-BIDIFAC+ to impute missing values following the probabilistic matrix factorization of \citet{mnih2007probabilistic}. 

We use the same two-stage adaptive procedure as in Section \ref{sec:agl}. In the first stage, we learn weights $\omega_{i,l}^{(U)}$ and $\omega_{j,l}^{(V)}$ by solving the observed-entry objective with a reduced penalty level, and then the final estimator is obtained by solving the same observed-entry objective with these weights. The resulting estimates $\widehat{\mathbf U}\widehat{\mathbf V}^\top$ are used to impute missing values. The tuning parameters are computed once from the initially observed data. Specifically, when estimating $\sigma_1(\bS)$ with missing data, we adjust for the observation rate $|\Omega|/(mn)$. 

\section{Theoretical analysis}\label{sec:theory}

For simplicity, we also assume that the full matrix is observed, i.e. $|\Omega| = mn$.
The objective in \eqref{eq:obj_uv} is non-convex jointly in $\bU$ and $\bV$, but is optimized by alternating between convex group lasso subproblems for $\bU$ and $\bV$. If one factor matrix is fixed at its population counterpart, the update of the other factor matrix reduces to a standard group lasso problem, for which support recovery conditions are well studied \citep{negahban2012unified, wainwright2019high}. Motivated by this observation, we analyze each alternating update locally by treating the current iterate, for example $\bU^{(t)}$, as a perturbed version of the population factor matrix. 

We analyze the update step $\bU$ for a fixed $\widetilde{\bV}$, which may not exactly factorize the underlying matrix $\bS$.
In particular, we allow the column dimension of $\widetilde{\bV}$, denoted by $d$, to exceed the rank of $\bS$, denoted by $r \leq d$.
We show that the group lasso problem can correctly prune the rank of the corresponding $\bU$ to match the module structure of an exact group sparse factorization. The update of $\bV$ given $\widetilde{\bU}$ and corresponding theoretical results are symmetric.

Given $\widetilde{\bV}$, the $\bU$-update can be written as a weighted group lasso problem:
\begin{equation} \label{eq:glasso}
\widehat{\bbeta}^U
=
\argmin_{\bbeta}~
\frac{1}{2}
||
\by -\bX_U(\widetilde{\bV})\bbeta
||_2^2
+
(mn) \lambda_U
\sum_{l=1}^{d}
\sum_{i=1}^p
\omega^{(U)}_{il} 
||
\bbeta_{i,l}
||_2
\end{equation}
where
$
\by
=
\operatorname{vec}(\bY),
\bX_U(\widetilde{\bV})
=
\widetilde{\bV}\otimes \bI_m,
\bbeta=\operatorname{vec}(\bU),
$ and the  $\bbeta_{i,l}$ corresponds to $\bu_{\mathcal G_i,l}$.

We will assume that, without loss of generality, the first $r$ columns of the factor matrix $\widetilde{\bV}$ are a perturbation of the true factor matrix $\bV^*$ for some representative factorization $\bS_*=\bU_*\bV_*^\top$ with group sparse columns denoted by
$
\bU_* = [\bU_*^{(1)}, \dots, \bU_*^{(K)}], \quad \bV_* = [\bV_*^{(1)}, \dots, \bV_*^{(K)}],
$
having minimal column dimension $r$.
In this one-sided problem, the target of optimization, including the group sparsity pattern, is fixed by the factor representation of $\widetilde{\bV}$.
Let
$
\mathcal A^*_U
=
\left\{
(i,l):\|(\bu_*)_{\mathcal G_i,l}\|_2>0
\right\}
$ denote the true active row-group support for the columns of $\bU_*$, and denote $\bP_{r,d} = (\bI_r , \bm{0}) \in \mathbb{R}^{r \times d}$. Then
\[
\bY
=
\bU_*\bV_*^\top+\bE
=
\bU_*\bP_{r,d}\widetilde{\bV}^{\top}
+
\bU_*\bP_{r,d}(\bV_*\bP_{r,d} -\widetilde{\bV})^\top
+
\bE,
\]
defining a coefficient vector $\bbeta^* = \vecz(\bU_*\bP_{r,d})$,  we can write
\begin{equation} \label{eq:error_decomp}
    \by = \bX_U(\widetilde{\bV})\bbeta^* + \vecop{\bE} + \{\bX_U(\bV_*\bP_{r,d}) - \bX_U(\widetilde{\bV})\} \bbeta^*.
\end{equation}
Thus, $\widetilde{\bV}$ contributes to an additional perturbation term in the group lasso problem, but additional spurious columns only extend the dimension of $\bbeta^*$ with zero entries, without affecting the sparsity pattern $\mathcal{A}^*_{U}$.
We assume the following sufficient conditions hold.
\begin{assumption}
\label{ass:local_U}
Assume that the following conditions are satisfied by $\widetilde{\bV}$.
\begin{enumerate}
    \item \textit{Design near-orthogonality.}
    The columns of $\widetilde{\bV}$ are nearly orthogonal in the sense that
    $
    \max_{l \neq l'} \lvert \tilde{\bv}_{l}^{\top} \tilde{\bv}_{l'} \rvert \leq \alpha \min_{l} (\tilde{\bv}_{l}^{\top} \tilde{\bv}_{l}) / r
    $
    for a constant $\alpha < 1/2$. 

    \item \textit{Design conditioning.}  $\min_{l} (\tilde{\bv}_{l}^{\top} \tilde{\bv}_{l}) \geq \kappa \lVert \widetilde{\bV} \rVert^2_{1,2}$ for a constant $\kappa \in (0,1)$.

    \item \textit{Subgaussian noise.}
    The entries of $\bE$ are mutually independent and subgaussian, with subgaussian variance proxy at most $\sigma^2$.

\end{enumerate}
\end{assumption}

We make near-orthogonality and conditioning assumptions on the current estimate $\widetilde{\bV}$, since this is the available object to check the conditions and calibrate the tuning parameter.

\begin{proposition}[One-step support recovery]
\label{prop:one_step_U}
Suppose Assumption \ref{ass:local_U} holds. Then the solution of the $\bU$-update recovers the true active group support of $\bU_*$
\[
\widehat{\mathcal A}_U
=
\left\{
(i,l):\|\hat{\bu}_{\mathcal G_i,l}\|_2>0
\right\}
=
\mathcal{A}^*_U.
\]
with probability at least $1 - \eta$, as long as $\lambda_U$ satisfies the two inequalities
\begin{align}
    \lambda_U &> \frac{3 \lVert \widetilde{\bV} \rVert_{1,2} }{mn(1-2\alpha)\omega_{\min}} \max\bigg\{ \frac{2 (1 - \alpha) \sigma}{c_E} \left\{ \sqrt{M\log(5)} + \sqrt{\log (2pd/\eta)} \right\},r C_{\beta_U} \| \bV_* - \widetilde{\bV} \bP_{r,d}^{\tp} ~ \|_{1,2} \bigg\}, \nonumber \\ \lambda_U &\leq \frac{c_{\beta_U} (1 - \alpha) \kappa \lVert \widetilde{\bV} \rVert^2_{1,2} }{mn(C_{\alpha} \omega_{\min} + \omega_{\max})}, \label{eq:tuning_range}
\end{align}
where $c_E>0$ is a universal constant, $C_{\beta_U} = \max_{\mathcal{A}^*_U} \lVert \bbeta_{i,l} \rVert_2$, $c_{\beta_U} = \min_{\mathcal{A}^*_U} \lVert \bbeta_{i,l} \rVert_2$, $\omega_{\max} = \max_{\mathcal{A}^*_U} \omega_{il}^{(U)}$, $\omega_{\min} = \min_{\mathcal{A}^{*,\comp}_U} \omega_{il}^{(U)}$ with $\omega_{\max} \leq \omega_{\min}$, $C_{\alpha}$ is a constant depending only on $\alpha$ in Assumption \ref{ass:local_U}.1, and $M = \max_{1 \leq i \leq p} m_i$ is the maximum group size. 
Moreover, the estimation error on the active groups satisfies
\begin{equation} \label{eq:glasso_error}
\max_{(i,l) \in \mathcal{A}^*_U} \|\widehat{\bbeta}^U_{i,l}
-
\bbeta_{i,l}^*
\|_2
\leq
 \frac{(C_{\alpha} \omega_{\min} + \omega_{\max}) mn \lambda_U}{\kappa(1 - \alpha) \lVert \widetilde{\bV} \rVert^2_{1,2}}.
\end{equation}
\end{proposition}

In practice, the near-orthogonality condition in Assumption \ref{ass:local_U}.1 may be too strong to hold, and Proposition \ref{prop:one_step_U} shows that support recovery requires stronger effective signals and a smaller perturbation of the current factor estimate, as the admissible range of the corresponding tuning parameter narrows with increasing $\alpha$.
{\em Mutual incoherence} conditions of this form are natural in lasso and group lasso problems, and are required so that the penalty terms can identify the correct sparse solutions, even in the noiseless case \citep{wainwright2019high}.
In this case, we write an interpretable sufficient condition in terms of the columns of the factor matrix $\widetilde{\bV}$, rather than an opaque condition on the full design matrix $\bX_U(\widetilde{\bV})$.

The score domination condition motivates the tuning parameter in the model. The group design for the $\bU$-update has block operator norm proportional to $\|\widetilde{\bv}_l\|_2$. Thus, for subgaussian noise with design perturbation satisfying $\| \bV_* - \widetilde{\bV} \bP_{r,d}^{\tp} ~ \|_{1,2} \ll \sigma \sqrt{M} / r C_{\beta_U}$, a sufficient order is
\[
\lambda_{U}
\asymp
\frac{\|\widetilde{\bV}\|_{1,2}}{mn}
\{
\sqrt{M} + \sqrt{\log(pd)}\}
,
\]
which is the usual rate required for a group lasso problem with design $\bX_U(\widetilde{\bV})$.
Under a balanced factorization, $\|\widetilde{\bV}\|_{1,2}$ is of order $\sqrt{\sigma_1(\bS)}$, motivating the tuning rule in \eqref{eq:lambda_used}. 
This valid choice of $\lambda_U$ has a particularly weak dependence on $d$, implying that subject to some orthogonality conditions, the rank of the factorization can be correctly reduced in the presence of spurious columns of $\widetilde{\bV}$. 
There is a non-standard tuning regime when $\| \bV_* - \widetilde{\bV} \bP_{r,d}^{\tp} ~ \|_{1,2} \gg \sigma \sqrt{M} / r C_{\beta_U}$ and the lower bound in \eqref{eq:tuning_range} is dominated by the design perturbation.
In this case, as long as 
$
    \| \bV_* - \widetilde{\bV} \bP_{r,d}^{\tp} ~ \|_{1,2} \ll c_{\beta_U} \|\widetilde{\bV}\|_{1,2} / r C_{\beta_U}
$ with high probability,
correct selection is still possible, but requires a larger penalty to remove the effect of design misspecification.

Finally, note that Proposition~\ref{prop:one_step_U} depends on the column norms of the design factors $\widetilde{\bV}$ and the target factors $\bbeta^*$. In \eqref{eq:tuning_range}, if the column norms of $\widetilde{\bV}$ are artificially increased by unbalancing the factorization, the valid tuning range will not be affected due to the corresponding effect on $c_{\beta_U}$, as long as the tuning parameter is appropriately rescaled.
Similarly, although the estimation error in \eqref{eq:glasso_error} can be  artificially decreased by unbalancing the factorization, this will lead to a corresponding increase in the error for $\bV$-recovery in the next alternating problem and will not affect error rates for recovery of the entries of $\bS$.

\section{Simulation studies}\label{sec:sim}

\subsection{Simulation design} \label{subsec:sim_design}

Our simulation design, including module specification in the 3$\times$3 bidimensional structure, is illustrated in Figure \ref{fig:simdesign}. To reduce the computational cost in fitting BIDIFAC+, we set $p=q=3$ with $m_1=100, m_2=200,m_3=100$ and $n_1=150,n_2=200, n_3=250$, with 6 active modules in total. For module $k$, we first generated latent factor matrices $\bU_k\in\mathbb{R}^{m \times r}$ and $\bV_k\in \mathbb{R}^{n\times r}$ with independent $\mathcal{N}(0,1)$ entries and then set the entries outside the module support to zero. We compare settings with module ranks $r \in \{1,5\}$. The $k$th module was then constructed as $\bS^{(k)} = c \bU_k \bV_k^\top$, where $c > 0$ is a module scale parameter controlling signal strength. The entries of $\bE$ are independent $\mathcal{N}(0,1)$. We compared settings with $c\in\lbrace 1/8, 2/8, 3/8, 4/8\rbrace$ to control the signal-to-noise ratio. $c=1/8$ is the case where the top singular values of $\bY$ are barely higher than $\sqrt{n}+\sqrt{m}$, a probabilistic upper bound of $\sigma_1(\bE)$, while $c=4/8$ gave a clear elbow pattern in the singular values of $\bY$.

We compare GL-BIDIFAC+ to (i) BIDIFAC+ with $K=4$ (underspecified), $K=6$ (correctly specified), or $K$ unspecified (i.e., overspecified), where the $K=4$ and $K=6$ versions use a module-searching step during estimation, and (ii) MOFA+ with different pruning levels applied to each subcolumn of the estimated feature loadings and subject scores. The rank for MOFA+ was set to be the number of singular values exceeding $\sqrt{n}+\sqrt{m}$, which is the same as that used for GL-BIDIFAC+. Also, we set $\hat{\bw}_{il}=\bzero$ for MOFA+ when  $||\hat{\bw}_{il}||_2/\sqrt{m_i}<\texttt{eps}$ and $\hat{\bz}_{jl}=\bzero$ when  $||\hat{\bz}_{jl}||_2/\sqrt{n_j}<\texttt{eps}$ with $\texttt{eps}\in\lbrace 0.01, 0.05, 0.1, 0.5\rbrace$.

\subsection{Signal and module recovery}\label{sec:sim1}
We consider several metrics to evaluate the performance. Since our primary goal is to recover the module-specific variations, we consider (i) relative squared error (RSE) and (ii) relative false magnitude (RFM) defined as follows. 
\begin{align*}
\text{RSE}=\frac{\sum_{k\in  \mathcal{K} } ||\widehat{\bS}^{(k)}-\bS^{(k)}||_F^2}{\sum_{k\in\mathcal{K}} ||\bS^{(k)}||_F^2}\qquad \text{and} \qquad
\text{RFM}=\frac{\sum_{k\notin \mathcal{K}} ||\widehat{\bS}^{(k)}||_F^2}{mn},
\end{align*}
where $\mathcal{K}$ denotes the set of active modules.
For RFM, the denominator corresponds to the expected total variation of the white noise.

We also consider the number of true positive  and false positive modules defined as follows.
\begin{align*}
\text{TP}_{\text{module}}=\sum_{k\in {\mathcal{K}}} I(\widehat{\bS}^{(k)}\neq \bzero)\qquad \text{and} \qquad
\text{FP}_{\text{module}}=\sum_{k\notin {\mathcal{K}}} I(\widehat{\bS}^{(k)}\neq \bzero).
\end{align*}

\begin{figure}[h]
    \centering
    \includegraphics[width=0.9\linewidth]{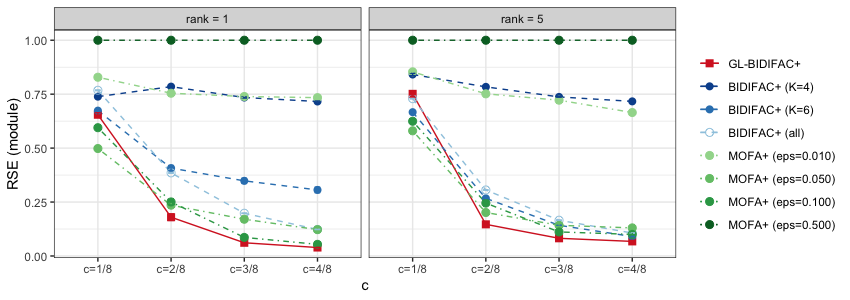}

    \includegraphics[width=0.9\linewidth]{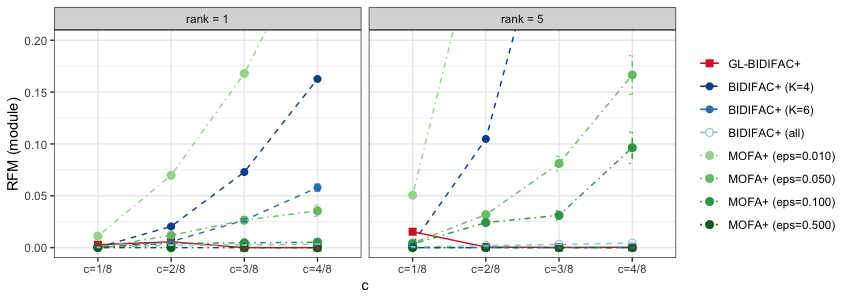}
    \caption{Summary of the simulation results for RSE and RFM.}
    \label{fig:RSE}
\end{figure}

\begin{figure}[h]
    \centering
    \includegraphics[width=0.9\linewidth]{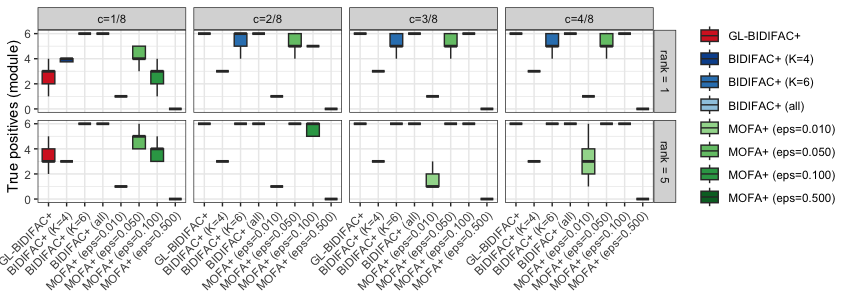}

    \includegraphics[width=0.9\linewidth]{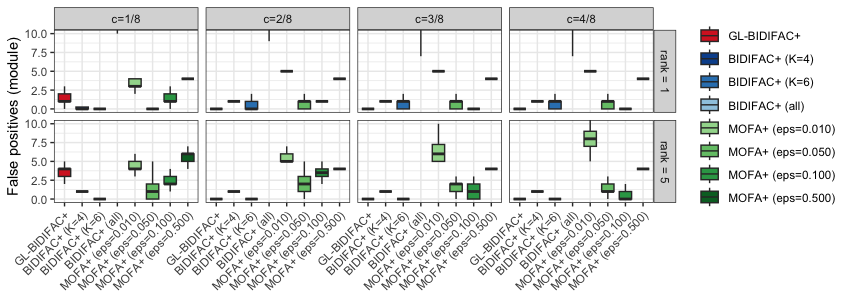}
    \caption{Summary of the simulation results for true/false positives. For BIDIFAC+ (all), the number of FP modules was usually above 10; the y-axis  is bounded by 10 for visualization.}
    \label{fig:TPFP}
\end{figure}

For each $c$ and $r$, the simulation was repeated 200 times, and the results are shown in Figures \ref{fig:RSE} and \ref{fig:TPFP}. Across all simulation settings, GL-BIDIFAC+ showed the most favorable overall performance. As the signal level increased from $c=1/8$ to $c=4/8$, its module-wise reconstruction error decreased rapidly and consistently for both rank settings, while maintaining high true-positive module recovery with very few false-positive modules. In particular, once the signal became moderate ($c \ge 2/8$), GL-BIDIFAC+ nearly recovered the full module structure in the rank-1 setting and remained clearly superior in the rank-5 setting as well. These gains were accompanied by near-zero relative false magnitude, indicating that the improved recovery was not driven by spurious extra modules but rather by more accurate identification of the true bidimensionally linked structure.

By contrast, BIDIFAC+ was substantially more sensitive to the choice of the maximum number of modules. When the model was underspecified ($K=4$), it systematically failed to recover all true modules, which resulted in larger reconstruction error and reduced true-positive counts. Increasing the allowed number of modules improved recovery, but the unrestricted version tended to admit more false structure, as reflected in elevated false-positive counts and larger relative false magnitude. Thus, while BIDIFAC+ could perform reasonably well when the module number was chosen correctly, its behavior was noticeably less stable than that of GL-BIDIFAC+, especially when the model complexity was misspecified.

MOFA+ exhibited a different pattern. Its performance depended strongly on the post hoc threshold \texttt{eps} used to define module supports from dense factor estimates. Smaller thresholds produced more aggressive support declarations, which sometimes improved true-positive recovery but also led to substantially larger numbers of false-positive modules and inflated false magnitude. Larger thresholds reduced false discoveries but often did so at the cost of worse reconstruction or missed modules. 

\subsection{Imputation performance}

Based on the previous simulation results, we compared GL-BIDIFAC+ to (i) BIDIFAC+ ($K=6$), (ii) BIDIFAC+ (all), (iii) MOFA+, as well as (iv) softImpute applied to $\bY$ (i.e., without considering the group structure) using the tuning parameter $\sqrt{n}+\sqrt{m}$ \citep{mazumder2010spectral}. We note that no post-pruning is needed for MOFA+ to impute missing values; therefore, we set $\texttt{eps}=0$, which is the default in the \texttt{MOFA2} R package.

Using the simulation design described in Section~\ref{subsec:sim_design}, we considered two scenarios for entry missingness. The first scenario had 10,000 entries (4.17\%) of $\bY$ missing at random. The second scenario had 5\% of missing rate. Following \citet{park2020integrative}, we used
\begin{align*}
\text{ImputeErr}
=
\frac{
\sum_{(r,s)\in \Omega^{\comp}}
| \bS[r,s] - \widehat{\bS}[r,s] |^2
}{
\sum_{(r,s)\in \Omega^{\comp}}
\left| \bS[r,s] \right|^2
}
\end{align*}
to evaluate the  performance,  where $\Omega^{\comp}$ is the set of indices of the missing values. 

\begin{figure}[t]
    \centering
    \includegraphics[width=0.9\linewidth]{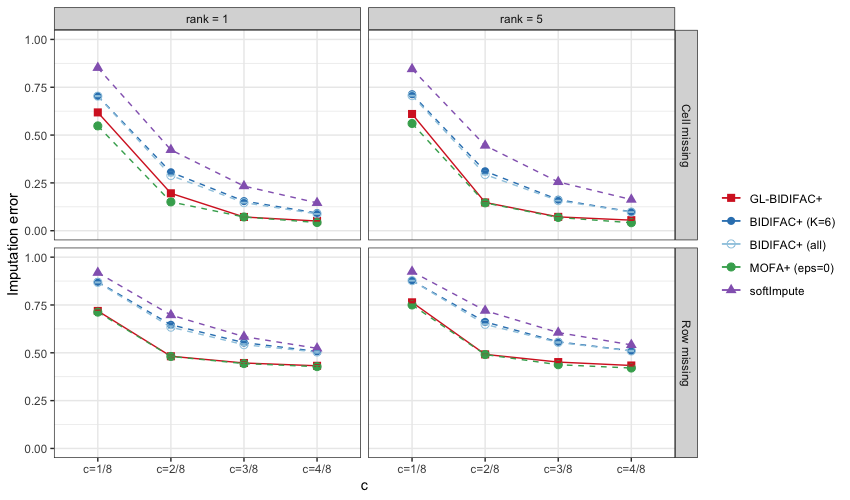}
    \caption{Summary of the imputation performance of GL-BIDIFAC+ and competitors.}
    \label{fig:imp}
\end{figure}

For each setting of $c$, $r$, and missingness structure, the simulation was repeated 200 times. Figure \ref{fig:imp} summarizes the imputation performance.  Across both missingness mechanisms, GL-BIDIFAC+ achieved imputation accuracy comparable to MOFA+ and consistently improved over BIDIFAC+ and softImpute. The advantage over softImpute indicates that exploiting the bidimensional group structure provides benefits beyond generic low-rank matrix completion, and the advantage over BIDIFAC+ can be partially attributed to (i) the difference between the nuclear norm penalty and the Schatten-$1/2$ penalty  (Proposition \ref{prop:shrinking}) and (ii) adaptive weighting in GL-BIDIFAC+ that BIDIFAC+ does not incorporate.  The slightly better performance of MOFA+ over GL-BIDIFAC+ suggests that further fine-tuning of GL-BIDIFAC+ may be needed; at the same time, it underscores that the tuning parameter selection in GL-BIDIFAC+ is designed based on support and module recovery in the asymptotic regime, not necessarily on global signal recovery.  In addition, from a probabilistic viewpoint, the spike-and-slab prior with ARD in MOFA+ may align more directly with the data-generating model than the multivariate Laplace prior implied by a Bayesian interpretation of the GL-BIDIFAC+ penalty term.  

\section{Analysis of TCGA data}\label{sec:data}

TCGA’s pan-cancer data have been extensively analyzed, and many of the dominant cancer-specific subtype, tissue-origin, and histology-related patterns are well documented \citep{weinstein2013cancer}. This makes the data useful for evaluating whether an unsupervised module-discovery method can recover known molecular structure while also providing a coherent decomposition into global, cancer-specific, and omics-restricted components.

\subsection{Preprocessing}

In our analysis, we considered four omics platforms available across all cancer types. Following \citet{lock2022bidimensional}, we used four omics sources: batch-corrected RNA-seq mRNA expression, batch-corrected miRNA-seq expression, between-platform normalized DNA methylation from the Illumina 27K and 450K arrays, and batch-corrected reverse-phase protein array data. RNA-Seq and miRNA-Seq counts were log-transformed $(\log(1+x))$, and each genomic feature was centered within cancer type to remove baseline differences across cancer types, and features were filtered to the 1000 genes and 1000 methylation CpG sites with the largest standard deviations after centering, leaving $m_1=1000$ genes, $m_2=743$ miRNAs, $m_3=1000$ CpGs, and $m_4=198$ protein features. In our analysis, we excluded cancer types with fewer than 100 tumors to avoid unstable estimation driven by very small sample sizes, which resulted in 21 cancer types. The number of samples in each tumor type varied between 115 (pancreatic adenocarcinoma; PAAD) and 848 (breast invasive adenocarcinoma; BRCA). Abbreviations of 21 cancers are summarized in the Supplementary Material.

\subsection{Findings}

The complete GL-BIDIFAC+ model fit required approximately 8 hours of wall-clock time on a high-performance computing node. It yielded 302 modules in total, with the ranks totaling 530.  This decomposition contained more modules than the BIDIFAC+ analysis in \citet{lock2022bidimensional}, but it remained  parsimonious at the module level, with most estimated modules having rank one. Our estimated signals  showed substantial rank reduction, as
$\text{rank}(\widehat{\bS}_{ij})<\text{rank}(\bY_{ij})$
for all $i,j$ and
$\text{rank}(\widehat{\bS}^{(k)})<
\min_{i,j}\{\text{rank}(\bY_{ij})\}$ for all $k$.
Also, the pairwise principal angles between  modules sharing at least one subblock were  close to $90^\circ$ in both the loading (median: 88.30$^\circ$) and score (median: 88.98$^\circ$) spaces, and only 1 pair from loading and 6 pairs from the score spaces had angles less than 45$^\circ$ (minimum: 31.57$^\circ$). 

The support, rank, and sum of squares of each module are  availabale at \url{https://github.com/junjypark/GLBP_repository} (\texttt{module\_summary\_table.xlsx}). Modules were numbered based on their sum of squares. To avoid treating the 302 modules as a list of unrelated findings, we categorize the modules by their support patterns, and make several interpretable findings which we enumerate below.

First, the modules explaining the largest amount of variation were predominantly pan-cancer or \textit{near}-pan-cancer modules active across all retained cancer types, e.g., miRNA only (module 1), protein only (module 2), mRNA + miRNA (module 3), mRNA only (module 4), mRNA + methylation (module 5), methylation only (module 6), and mRNA + miRNA + methylation (module 7). This is expected because these modules accumulate variation over the largest submatrices and capture broad pan-cancer molecular variation, together with platform-wide covariance patterns. Compared with the BIDIFAC+ analysis, in which broad pan-cancer variation was decomposed primarily into all-omics, miRNA-only, and mRNA-only modules, GL-BIDIFAC+ produced a sharper separation of shared omics structures.

\begin{figure}
    \centering
    \includegraphics[width=0.9\linewidth]{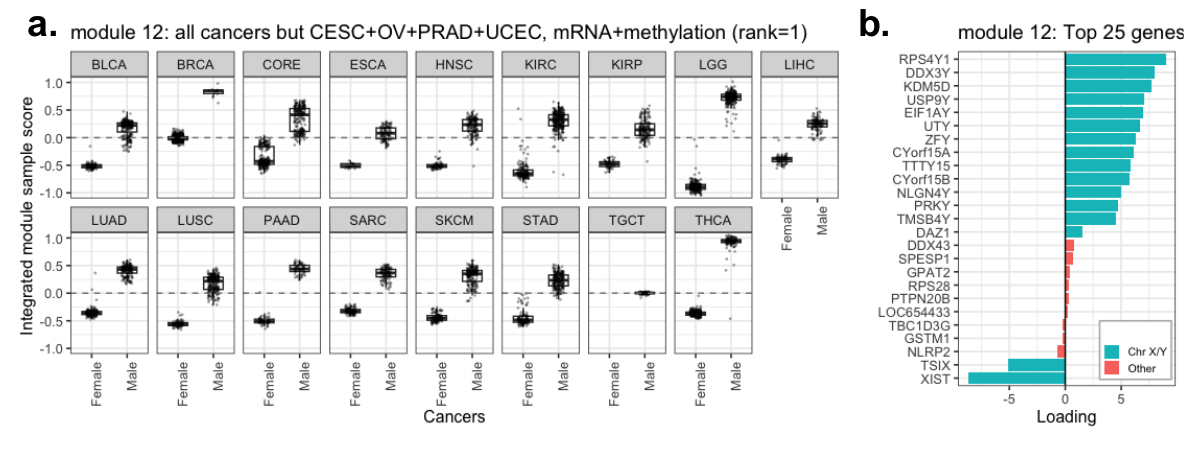}
    \caption{Summary of module 12.
(a) Distribution of  subject scores by sex within each cancer type. 
(b) The 25 genes with the largest absolute mRNA loadings.
}
    \label{fig:sex}
\end{figure}

Second, among several near-pan-cancer modules, module 12 had  all cancers except for four sex-specific cancers (CESC, OV, PRAD, and UCEC) across mRNA + methylation, and it captured a strong sex-associated variation. Among cancer types represented by both sexes (i.e., except for TGCT that showed tiny variation), scores consistently separated female and male samples (Figure \ref{fig:sex}(a)). This interpretation was supported by the mRNA loadings, where the largest positive loadings were dominated by genes on chromosomes X/Y (Figure \ref{fig:sex}(b)). This module closely parallels the sex-associated module identified by BIDIFAC+, but isolates the corresponding variation within the mRNA + methylation structure.
Several near-pan-cancer modules (e.g., modules 9, 14, 16, 17, 19, 28, and 43) often excluded  esophageal carcinoma (ESCA), pancreatic adenocarcinoma (PAAD), and testicular germ cell tumors (TGCT). Although these exclusions should be interpreted carefully, such a pattern is consistent with features of these cancer types. For example, TGCT is a germ-cell tumor rather than an epithelial carcinoma; ESCA combines molecularly distinct adenocarcinoma and squamous histologies under a single cancer-type label \citep{cancer2017integrated}; and PAAD often has low neoplastic cellularity, so bulk molecular profiles may contain substantial stromal or desmoplastic contributions that weaken their alignment with broad tumor-intrinsic pan-cancer variation \citep{raphael2017integrated}.

\begin{figure}[t]
    \centering
    \includegraphics[width=0.9\linewidth]{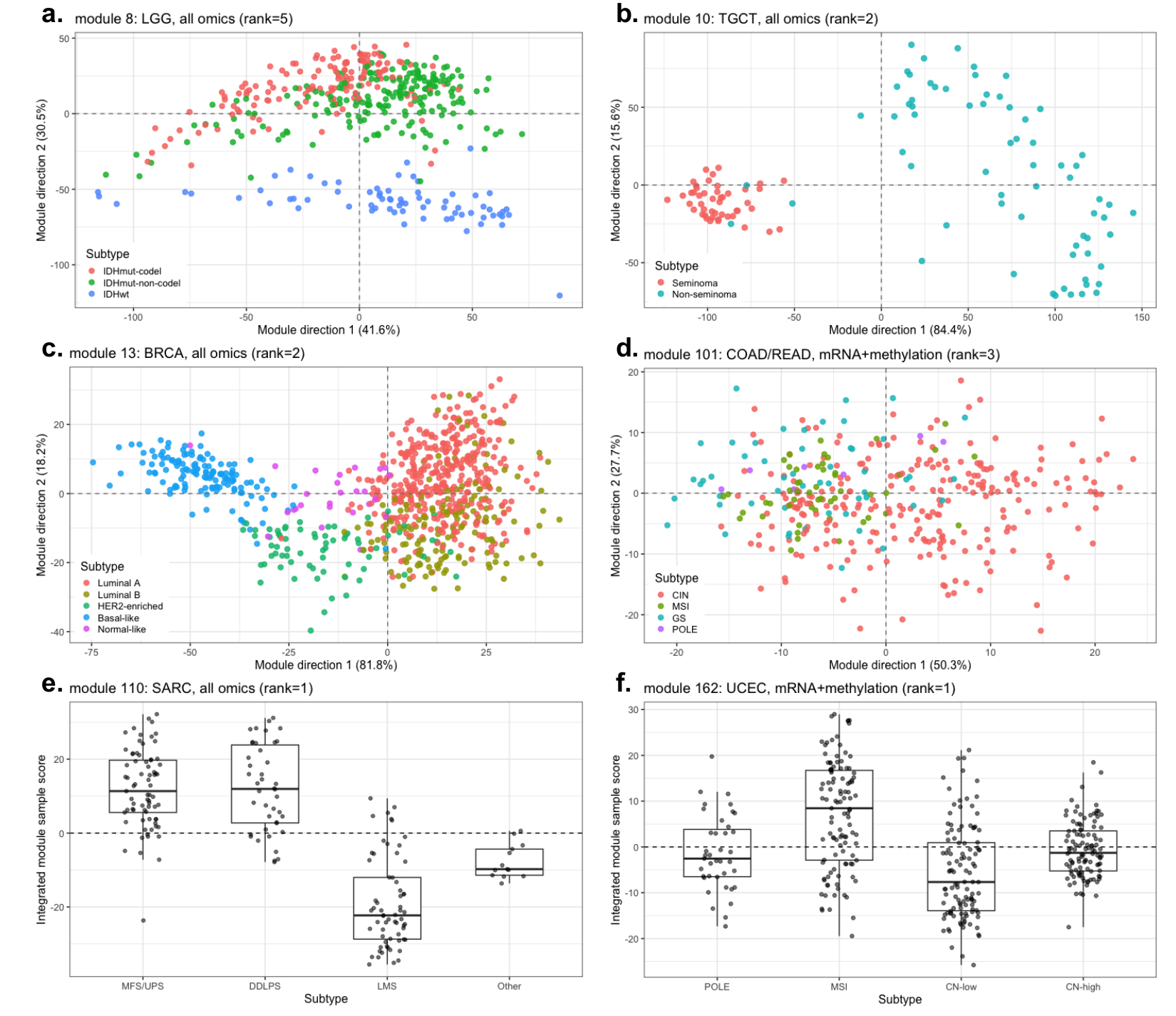}
    \caption{
    Visualization of six modules  that are specific to one cancer type and explained by established molecular or histologic subtypes from TCGA.  
    }
    \label{fig:subtype}
\end{figure}

Third, there were a number of cancer-specific modules including low-grade glioma (LGG), TGCT, BRCA, thyroid carcinoma (THCA), PAAD, kidney renal clear cell carcinoma (KIRC), sarcoma (SARC), and others. Among these, we focused on cancers where molecular or histologic subtypes identified by TCGA, and examined whether variabilities are explained by known subtypes. The visualizations of these modules are provided in Figure \ref{fig:subtype}. Module 8 for LGG across all omics recovered the major IDH/1p19q-defined subtype structure, separating IDH-wildtype, IDH-mutant/non-codeleted, and IDH-mutant/1p19q-codeleted tumors (Figure \ref{fig:subtype}(a)) \citep{cancer2015comprehensive}. Module 10 for TGCT across all omics clearly separated seminoma from non-seminoma, consistent with the heterogeneity of this group (Figure \ref{fig:subtype}(b)) \citep{shen2018integrated}. Module 13 for BRCA across all omics showed a clear separation of  basal-like tumors from the other BRCA subtypes, while the remaining subtypes exhibited greater overlap in the two-dimensional module-score space (Figure \ref{fig:subtype}(c)) \citep{CancerGenomeAtlasNetwork2012}.  Module 101 for COAD/READ across mRNA and methylation showed partial organization of the CIN, MSI, GS, and POLE molecular classes, with the clearest contrast occurring between GS tumors and the subset of CIN tumors extending toward positive scores along the first module direction (Figure \ref{fig:subtype}(d)). Rank-one modules also captured interpretable subtype contrasts. Module 162 for uterine corpus endometrial carcinoma (UCEC) across mRNA and methylation captured an epigenetic axis that most clearly distinguished MSI tumors from copy-number-low tumors, while POLE and copy-number-high tumors occupied intermediate score ranges (Figure \ref{fig:subtype}(f)) \citep{tcga2013}. Module 110 for SARC across all omics distinguished leiomyosarcoma from MFS/UPS and dedifferentiated liposarcoma, with the leiomyosarcoma samples exhibiting markedly lower module scores (Figure \ref{fig:subtype}(e)) \citep{lazar2017comprehensive}. In the BIDIFAC+ analysis, subtype-associated variation was also observed in BRCA and UCEC; however, LGG subtype variation was primarily captured by a methylation-only module, and TGCT-specific subtype structure was not identified.

Fourth, GL-BIDIFAC+ identified  modules shared across groups of biologically related cancers. These modules often reflected similarity in tissue of origin or histology. For example, modules 33, 69, 75 and  100 captured variation common to kidney cancers (KIRC and kidney renal papillary cell carcinoma [KIRP]). Similarly, modules 164  and 173  suggested gastrointestinal variations  (COAD/READ and stomach adenocarcinoma [STAD]). Modules 77 and 134  represented  repeatedly identified  BRCA and prostate adenocarcinoma (PRAD) that may reflect hormone-responsive epithelial variation. Module 47 (ovarian serous cystadenocarcinoma [OV] + UCEC,
mRNA + miRNA + methylation) identified  shared variations across  gynecologic cancers arising in the female reproductive tract. Its latent directions were predominantly driven by UCEC copy-number-high tumors and ovarian tumors, consistent with existing literature on a  relationship between serous-like endometrial and ovarian molecular variation \citep{tcga2013}.

\begin{figure}[!]
    \centering
    \includegraphics[width=0.9\linewidth]{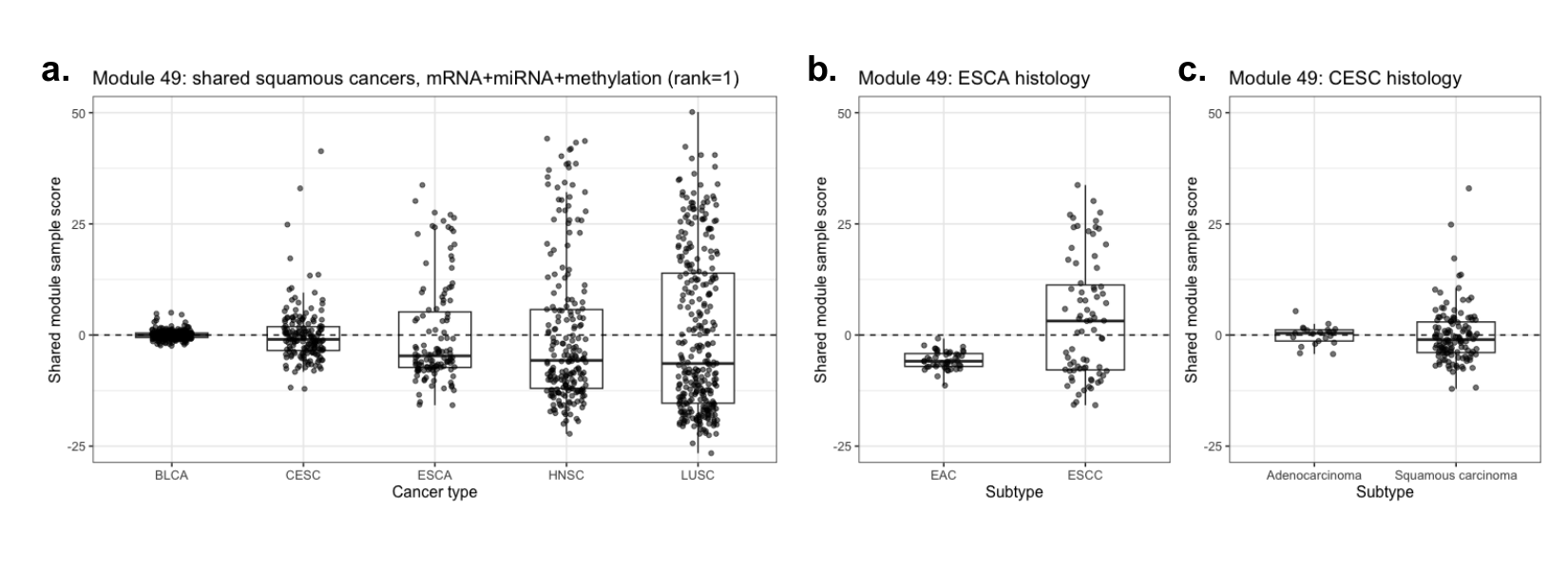}
    \caption{
    Summary of module 49.
    (a) Distribution of subject scores by the  cancer types.
    (b) Scores for esophageal adenocarcinoma (EAC) and esophageal squamous cell carcinoma (ESCC).
    (c) Scores for cervical and cervical squamous cell adenocarcinomas.    }
    \label{fig:epithelial}
\end{figure}

Fifth, module 49 (BLCA, CESC, ESCA, HNSC, and LUSC across mRNA + miRNA + methylation) captured cross-tissue variation associated with squamous differentiation program. From Figure \ref{fig:epithelial}(a), HNSC and LUSC, defined predominantly by squamous histology, showed higher variability than the others that consist of adenocarcinomas and squamous carcinomas. We therefore used the within-cancer histologic contrasts available in ESCA and CESC to assess whether variation of the module  was associated with squamous tumors. In ESCA, ESCC samples exhibited substantially greater variation (Figure \ref{fig:epithelial}(b)), and cervical squamous carcinomas showed considerably greater variation than adenocarcinomas (Figure \ref{fig:epithelial}(c)). Among the top mRNA loadings were genes known to be involved in squamous epithelial/keratinocyte differentiation, including KRT4, KRT13/14/17, SPRR1B, SPRR2A, SPRR2D, SPRR3, IVL, SCEL, PADI1, PADI3, and S100A8/S100A9. 

Taken all together, the analysis suggests that GL-BIDIFAC+ can distinguish broad pan-cancer variation from more localized biological signals. This separation is useful for interpretation, where large global modules summarize common molecular variability across the pan-cancer cohort, cancer-specific modules recover established molecular and histologic subtypes, and restricted modules reveal structured variation shared by select subsets of cancers.

\section{Discussion}\label{sec:discussion}

We have proposed GL-BIDIFAC+, a group-regularized matrix factorization method for decomposing bidimensionally linked matrices into low-rank modules. Compared to the existing methods, our method (i) significantly enhances computational scalability compared to BIDIFAC+, (ii) does not rely on post-hoc pruning (which might be arbitrary for MOFA+) or \textit{a priori} choice of the maximum number of modules (which might be arbitrary for BIDIFAC+), (iii) improves true and false positive rates for module recovery; and (iv) allows for model-based imputation. When GL-BIDIFAC+ was applied to the TCGA data, the recovered modules included broad pan-cancer modules, cancer-specific subtype modules, and modules shared by selected cancer types that are consistent with established TCGA biology.

GL-BIDIFAC+ is  suited for modeling with the Gaussian likelihood, and extensions to binary or count data would require more advanced modeling strategies such as those of \citet{li2018general}. We note, however, that existing methods that support modeling binary/count data (e.g., MOFA+) acknowledge that estimation may be less accurate than under a Gaussian model. Similarly, we also recommend normalization and variance-stabilizing transformations whenever possible followed by Gaussian modeling.

\section*{Software}

An R package for  GL-BIDIFAC+ is available at 
\ifblind
[blinded].
\else
\url{https://github.com/junjypark/GLBP}.
\fi

\section*{Disclosure statement}

The authors report there are no competing interests to declare. 

\ifblind
\else
\section*{Funding}

This work was supported by the McLaughlin Centre's Accelerator Grant (JYP) and the Natural Sciences and Engineering Research Council of Canada (RGPIN-2022-04831 for JYP and RGPIN-2025-02892 for PWM). The computing resources were enabled in part by support provided by University of Toronto and the Digital Research Alliance of Canada.
\fi

\section*{Data availability statement}

All data used in Section \ref{sec:data} are publicly available from the Cancer Genome Atlas. Codes to preprocess the data are provided at 
\ifblind
[blinded].
\else
\url{https://github.com/junjypark/GLBP_repository}.
\fi

\bibliographystyle{apalike}
{
\bibliography{references}
}

\newpage

\begin{center}
    {\LARGE \textbf{Supplementary materials}}
\end{center}
\appendix

\renewcommand{\thefigure}{S\arabic{figure}}
\setcounter{figure}{0}

\section{Additional information on Section \ref{sec:data}}

\subsection{Acronyms of the cancers}

\begin{table}[htbp]
    \centering
    \caption{Acronyms of the 21 cancer types included in the TCGA analysis.}
    \label{tab:cancer_abbreviations}
    \small
    \begin{tabular}{ll}
        \toprule
        Abbreviation & Cancer type \\
        \midrule
        BLCA & Bladder urothelial carcinoma \\
        BRCA & Breast invasive carcinoma \\
        CESC & Cervical carcinoma \\
        CORE (COAD/READ) & Colorectal adenocarcinoma \\
        ESCA & Esophageal carcinoma \\
        HNSC & Head and neck squamous cell carcinoma \\
        KIRC & Kidney renal clear cell carcinoma \\
        KIRP & Kidney renal papillary cell carcinoma \\
        LGG & Brain lower-grade glioma \\
        LIHC & Liver hepatocellular carcinoma \\
        LUAD & Lung adenocarcinoma \\
        LUSC & Lung squamous cell carcinoma \\
        OV & Ovarian serous cystadenocarcinoma \\
        PAAD & Pancreatic adenocarcinoma \\
        PRAD & Prostate adenocarcinoma \\
        SARC & Sarcoma \\
        SKCM & Skin cutaneous melanoma \\
        STAD & Stomach adenocarcinoma \\
        TGCT & Testicular germ cell tumors \\
        THCA & Thyroid carcinoma \\
        UCEC & Uterine corpus endometrial carcinoma \\
        \bottomrule
    \end{tabular}
\end{table}

\newpage

\subsection{Summary of the ranks of the signals}

\begin{table}[H]
    \centering
    \caption{
    Sum of the ranks of the estimated modules contributing to each
    cancer-by-omics submatrix. Each entry is the sum of the ranks of
    all modules whose support includes both the corresponding cancer
    type and omics platform. 
    }
    \label{tab:subblock_rank_sums}
    \small
    \setlength{\tabcolsep}{5pt}
    \begin{tabular}{lrrrr}
        \toprule
        Cancer type
        & mRNA
        & miRNA
        & Methylation
        & Protein \\
        & ($p_1=1000$)
        & ($p_2=743$)
        & ($p_3=1000$)
        & ($p_4=198$) \\
        \midrule
        BLCA ($n_1=338$)   & 85  & 145 & 57 & 78 \\
        BRCA ($n_2=849$)   & 126 & 163 & 84 & 84 \\
        CESC ($n_3=170$)   & 64  & 103 & 38 & 61 \\
        CORE ($n_4=447$)   & 95  & 138 & 57 & 76 \\
        ESCA ($n_5=123$)   & 59  & 75  & 38 & 46 \\
        HNSC ($n_6=210$)   & 64  & 107 & 45 & 66 \\
        KIRC ($n_7=446$)   & 90  & 139 & 52 & 79 \\
        KIRP ($n_8=210$)   & 66  & 122 & 37 & 79 \\
        LGG ($n_9=425$)    & 93  & 151 & 51 & 87 \\
        LIHC ($n_{10}=175$) & 51 & 105 & 37 & 77 \\
        LUAD ($n_{11}=356$) & 84 & 145 & 55 & 80 \\
        LUSC ($n_{12}=304$) & 81 & 134 & 54 & 73 \\
        OV ($n_{13}=238$)   & 70 & 127 & 41 & 71 \\
        PAAD ($n_{14}=115$) & 53 & 70  & 33 & 50 \\
        PRAD ($n_{15}=347$) & 91 & 138 & 57 & 82 \\
        SARC ($n_{16}=217$) & 56 & 117 & 36 & 76 \\
        SKCM ($n_{17}=337$) & 80 & 143 & 58 & 75 \\
        STAD ($n_{18}=333$) & 85 & 138 & 56 & 77 \\
        TGCT ($n_{19}=118$) & 55 & 81  & 44 & 52 \\
        THCA ($n_{20}=370$) & 86 & 139 & 30 & 84 \\
        UCEC ($n_{21}=417$) & 86 & 147 & 57 & 79 \\
        \bottomrule
    \end{tabular}
\end{table}

\newpage
\section{Technical proofs}

\subsection{Proof of Proposition \ref{prop:ident}}

\begin{proof}
Fix a pair $\mathcal{A} \subseteq [p]$ and $\mathcal{B} \subseteq [q]$.
Note that for any $\mathcal{A}$ the set of row (resp.~column) modules decomposes as the disjoint union
$$
    \{\mathcal{A}' : \mathcal{A}' \subseteq \mathcal{A}\} \cup \{\mathcal{A}'' : \mathcal{A}'' \cap \mathcal{A}^{\comp} \neq \emptyset \},
$$
denote these sets by $x(\mathcal{A})$ and $y(\mathcal{A})$, and use analogous notation for $x(\mathcal{B})$ and $y(\mathcal{B})$. 
Thus, by \eqref{ident_cond}
\begin{align}
    &\operatorname{rank}(\bS) - \operatorname{rank}(\bS_{\mathcal{A}^{\comp} [q]}) - \operatorname{rank}(\bS_{[p]\mathcal{B}^{\comp}}) + \operatorname{rank}(\bS_{\mathcal{A}^{\comp}\mathcal{B}^{\comp}}) \nonumber \\
    = &\sum_{\mathcal{A}'}\sum_{\mathcal{B}'} r(\mathcal{A}',\mathcal{B}') - \sum_{\mathcal{A}' \in y(\mathcal{A})} \sum_{\mathcal{B}'} r(\mathcal{A}',\mathcal{B}') - \sum_{\mathcal{A}'} \sum_{\mathcal{B}' \in y(\mathcal{B})} r(\mathcal{A}',\mathcal{B}') + \sum_{\mathcal{A}' \in y(\mathcal{A})} \sum_{\mathcal{B}' \in y(\mathcal{B})} r(\mathcal{A}',\mathcal{B}') \nonumber \\
    = &\sum_{\mathcal{A}' \in x(\mathcal{A})} \sum_{\mathcal{B}' \in x(\mathcal{B})} r(\mathcal{A}',\mathcal{B}') \label{ident_claim}
\end{align}

Consider a partial order on modules such that $(\mathcal{A}',\mathcal{B}')$ is smaller than $(\mathcal{A},\mathcal{B})$ when $\mathcal{A}' \subseteq \mathcal{A}$ and $\mathcal{B}' \subseteq \mathcal{B}$.
The singleton modules with $\mathcal{A} = \{i_*\}$ and $\mathcal{B} = \{j_*\}$, are minimal elements as they contain no other modules. 
Thus, by \eqref{ident_claim}
$$
    r(\{i_*\},\{j_*\}) = \operatorname{rank}(\bS) - \operatorname{rank}(\bS_{\{i_*\}^{\comp} [q]}) - \operatorname{rank}(\bS_{[p]\{j_*\}^{\comp}}) + \operatorname{rank}(\bS_{\{i_*\}^{\comp}\{j_*\}^{\comp}})
$$ 
can be recovered uniquely from the ranks of the matrix blocks.
The same is true for larger modules, where $r(\mathcal{A},\mathcal{B})$ is recovered recursively from \eqref{ident_claim}, based on the ranks of the matrix blocks, and the module ranks $r(\mathcal{A}',\mathcal{B}')$ for all smaller modules.
\end{proof}

\subsection{Proof of Proposition \ref{prop1}}

\begin{proof}
Suppose without loss of generality that $\widehat{\bu}_{:,l}=\bzero$. We show that
$\widehat{\bv}_{:,l}=\bzero$.  Suppose, to the contrary, that
$\widehat{\bv}_{:,l}\neq\bzero$, and, for $t\in[0,1]$, define
\[
\bV(t)
=
\widehat{\bV}
-
t\widehat{\bv}_{:,l}\be_l^\top,
\]
where $\be_l$ is the $l$th standard basis vector in $\mathbb{R}^d$.
Because $\widehat{\bu}_{:,l}=\bzero$,
\[
\widehat{\bU}\bV(t)^\top
=
\widehat{\bU}\widehat{\bV}^\top
-
t\widehat{\bu}_{:,l}\widehat{\bv}_{:,l}^\top
=
\widehat{\bU}\widehat{\bV}^\top.
\]
Thus, replacing $\widehat{\bv}_{:,l}$ by
$(1-t)\widehat{\bv}_{:,l}$ does not change the loss function or the
penalty on $\bU$. On the other hand, the penalty on $\bV$ satisfies
\[
R\{\bV(t)\}
=
R(\widehat{\bV})
-
t\sum_{j=1}^q
\omega^{(V)}_{j,l}
\left\|\widehat{\bv}_{H_j,l}\right\|_2.
\]
Since the groups $\{H_1,\ldots,H_q\}$ form a partition,
$\widehat{\bv}_{:,l}\neq\bzero$ implies
\[
\sum_{j=1}^q
\omega^{(V)}_{j,l}
\left\|\widehat{\bv}_{H_j,l}\right\|_2
>0.
\]
Therefore, the objective strictly decreases for every sufficiently
small $t>0$, contradicting stationarity of
$(\widehat{\bU},\widehat{\bV})$. Hence
$\widehat{\bv}_{:,l}=\bzero$.
The proof of the reverse implication is identical.

The second statement follows immediately because the outer product of
two vectors is nonzero if and only if both vectors are nonzero.
\end{proof}

\subsection{Proof of Proposition \ref{prop:one_step_U}}

\begin{proof}
    The optimality conditions for \eqref{eq:glasso} are
$$
    \bzero = \XV^{\tp}\{\XV \widehat{\bbeta}^U - \by \} + (mn) \lambda_U \hat{\bz},
$$
where $\hat{\bz}\in\mathbb R^{md}$ is a vector of subgradients satisfying
\[
\hat{\bz}_{i,l}=
\begin{cases}
\omega_{i,l}^{(U)} \widehat{\bbeta}^U_{i,l}/\|\widehat{\bbeta}^U_{i,l}\|_2,
& \widehat{\bbeta}^U_{i,l}\ne 0,\\
\|\hat{\bz}_{i,l}\|_2\le \omega_{i,l}^{(U)},
& \widehat{\bbeta}^U_{i,l}=0.
\end{cases}
\]
Strict dual feasibility for inactive groups will be verified below.

Denote $\bV'_* = \bV_* \bP_{r,d}$. By \eqref{eq:error_decomp},
\begin{align*}
    \bzero = &\XV^{\tp} \XV \{ \hat{\bbeta}^U - \bbeta^* \} - \XV^{\tp} \left(\{\XVstar - \XV\} \bbeta^* + \vecz(\bE)\right) + (mn) \lambda_U \hat{\bz}.
\end{align*}
Partition the columns of $\XV$ according to the active and inactive groups:
$$
    \XVstar = [
        \XAstar, \XAcstar
    ], \quad \XV = [
        \XA, \XAc]. 
$$
Note that by construction, the spurious columns in $\widetilde{\bV}$ will all correspond to inactive groups in $\bbeta^*$.
Denote
$$
    \hat{\bm{\xi}} = \{\XVstar - \XV\} \bbeta^* + \vecz(\bE) = (\XAstar - \XA) \bbeta^*_{\mathcal{A}^*_U} + \vecz(\bE).
$$
For notational simplicity, we will also use $\widehat{\bX}_{i,l}$ to denote the columns of $\XV$ corresponding to group $(i,l)$. 
Using the Kronecker structure of $\XV$, note that
\begin{equation} \label{grams}
    \widehat{\bX}_{i,l}^{\tp} \widehat{\bX}_{i',l'} = \begin{cases}
        (\bv_{l}^{\tp} \bv_{l'}) \bI_{m_i}, \quad &i = i' \\
        \bzero, \quad &i \neq i'.
    \end{cases}
\end{equation}
Suppose that 
$$
    \hat{\bbeta}_{\mathcal{A}^*_U} =
\argmin_{\bbeta}~\left\{
\frac{1}{2}
||
\by -\XA \bbeta
||_2^2
+
(mn) \lambda_U
\sum_{(i,l) \in \mathcal{A}^*_U}
\omega^{(U)}_{il} 
||
\bbeta_{i,l}
||_2,
\right\}$$
which satisfies an optimality condition
\begin{equation} \label{eq:optimality_sub}
    \bzero = \XA^{\tp} \XA ( \widehat{\bbeta}_{\mathcal{A}^*_U} - \bbeta^*_{\mathcal{A}^*_U}) - \XA^{\tp} \widehat{\bm{\xi}} + (mn) \lambda_U \hat{\bz}_{\mathcal{A}^*_U}.
\end{equation}
Since $\XA$ has full rank,
\begin{equation} \label{eq:optimality_beta}
    \widehat{\bbeta}_{\mathcal{A}^*_U} - \bbeta^*_{\mathcal{A}^*_U} = ( \XA^{\tp} \XA )^{-1} ( \XA^{\tp} \widehat{\bm{\xi}} - (mn) \lambda_U \hat{\bz}_{\mathcal{A}^*_U} ).
\end{equation}
Define the vector
$$
    \hat{\bz}_{\mathcal{A}^{\comp}_U} = - \frac{1}{mn \lambda_U} \XAc^{\tp} \XA (\widehat{\bbeta}_{\mathcal{A}^*_U} - \bbeta^*_{\mathcal{A}^*_U}) + \frac{1}{mn\lambda_U} \XAc^{\tp} \widehat{\bm{\xi}}.
$$
It is straightforward to verify that with the subgradient vector $(\hat{\bz}_{\mathcal{A}^*_U}, \hat{\bz}_{\mathcal{A}^{\comp}_U})^{\tp}$, $(\widehat{\bbeta}_{\mathcal{A}^*_U}, \bzero)^{\tp}$ satisfies the optimality conditions of the original group lasso problem \eqref{eq:glasso}.
Thus a sufficient condition for $\widehat{\mathcal{A}}_U \subseteq \mathcal{A}_U^*$ is
\begin{equation} \label{subgrad_min}
    \max_{(i,l) \in \mathcal{A}_U^{*,\comp}} \| \hat{\bz}_{i,l} \|_2 < \omega_{\min},
\end{equation}
since \eqref{subgrad_min} implies that $\| \hat{\bz}_{i,l} \|_2 < \omega_{i,l}^{(U)}$ for all $(i,l) \notin \mathcal{A}^*_U$.
\begin{align*}
    \hat{\bz}_{i,l} =& - \frac{1}{mn \lambda_U} \hbX_{i,l}^{\tp} \XA (\hat{\bbeta}_{\mathcal{A}^*_U} - \bbeta^*_{\mathcal{A}^*_U}) + \frac{1}{mn\lambda_U} \hbX_{i,l}^{\tp} \hat{\bm{\xi}} \\
    =& - \frac{1}{mn \lambda_U} \hbX_{i,l}^{\tp} \XA ( \XA^{\tp} \XA )^{-1} ( \XA^{\tp} \hat{\bm{\xi}} - (mn) \lambda_U \hat{\bz}_{\mathcal{A}^*_U} ) + \frac{1}{mn\lambda_U} \hbX_{i,l}^{\tp} \hat{\bm{\xi}} \\
    =& \hbX_{i,l}^{\tp} \XA ( \XA^{\tp} \XA )^{-1} \hat{\bz}_{\mathcal{A}^*_U} + \frac{1}{mn\lambda_U} \hbX_{i,l}^{\tp} \widehat{\bP}^{\perp}_{\mathcal{A}^*_U} \hat{\bm{\xi}} \\
    =& \underbrace{\hbX_{i,l}^{\tp} \XA ( \XA^{\tp} \XA )^{-1} \hat{\bz}_{\mathcal{A}^*_U}}_{\text{(I)}} \\
    &+ \underbrace{\frac{1}{mn\lambda_U} \hbX_{i,l}^{\tp} \widehat{\bP}^{\perp}_{\mathcal{A}^*_U} \vecz(\bE)}_{\text{(II)}} \\
    &+ \underbrace{\frac{1}{mn\lambda_U} \hbX_{i,l}^{\tp} (\XAstar - \XA) \bbeta^*_{\mathcal{A}^*_U}}_{\text{(III)}} \\
    &- \underbrace{\frac{1}{mn\lambda_U} \hbX_{i,l}^{\tp} \XA ( \XA^{\tp} \XA )^{-1} \XA^{\tp} (\XAstar - \XA) \bbeta^*_{\mathcal{A}^*_U}}_{\text{(IV)}},
\end{align*}
where $\widehat{\bP}^{\perp}_{\mathcal{A}^*_U}$ is the projection operator onto the orthogonal complement of $\operatorname{col}(\XA)$.

\paragraph{Term (I)}
 Using \eqref{grams}, we see that $\hbX_{i,l}^{\tp} \XA=\bzero$ unless the active group is $(i,l')$ for some $l'$, and $\XA^{\tp} \XA$ is block diagonal across different $i$. Denote
$$
    \mathcal{A}^*_i = \{(i,l') \in \mathcal{A}^*_U : l'\in [r]\},
$$
and $\bP_{i,l}$ the projection operator onto the coordinates of group $(i,l)$.
Note that $\mathcal{A}^*_i$ only considers $l' \leq r$, as the spurious columns in $\widetilde{\bV}$ will correspond to inactive groups.
Finally, note that the Gram matrix $\XAi^{\tp} \XAi$ satisfies the diagonal dominance conditions of Lemma~\ref{lem:block_inverse} with $\gamma^* = 1/\min_l (\hat{\bv}_l^{\tp}\hat{\bv}_l)$.
Then,
\begin{align*}
    \| \hbX_{i,l}^{\tp} \XA ( \XA^{\tp} \XA )^{-1} \hat{\bz}_{\mathcal{A}^*_U}\|_2 &= \| \hbX_{i,l}^{\tp} \XAi ( \XAi^{\tp} \XAi )^{-1} \hat{\bz}_{\mathcal{A}^*_i} \|_2 \\
    &\leq \sum_{l' \neq l, l' \leq r} \| \hbX_{i,l}^{\tp} \XAi ( \XAi^{\tp} \XAi )^{-1} \bP_{i,l'} \hat{\bz}_{\mathcal{A}^*_i} \|_2 \\
    &\leq \omega_{\max} \sum_{l' \neq l, l' \leq r} \| \hbX_{i,l}^{\tp} \XAi ( \XAi^{\tp} \XAi )^{-1} \bP_{i,l'} \|_2 \\
    &\leq \frac{\omega_{\max} \gamma^*r}{1 - \alpha} \max_{l' \neq l} \lVert \hbX_{i,l}^{\tp} \XAi \bP_{i,l'} \|_2 \\
    &= \frac{\omega_{\max} \gamma^*r}{1 - \alpha} \max_{l' \neq l} \lVert \hbX_{i,l}^{\tp} \hbX_{i,l'} \|_2 \\
    &\leq \omega_{\max} \frac{\alpha}{1 - \alpha} \leq \omega_{\min} \frac{\alpha}{1 - \alpha},
\end{align*}
where the third-to-last inequality uses Lemma~\ref{lem:block_inverse}, and the second-to-last uses Assumption~\ref{ass:local_U}.1.

\paragraph{Term (II)} By standard covering arguments for Euclidean unit balls \citep[][Corollary 4.2.13]{vershynin2018high},
$$
    \max_{i,l} \left\lVert \frac{\hbX_{i,l}^{\tp} \widehat{\bP}^{\perp}_{\mathcal{A}^*_U} \vecz(\bE)}{mn \lambda_U} \right\rVert_2 \leq \max_{i,l} \max_{j \in \mathcal{N}_i} \frac{2 \bu_j^{\tp} \hbX_{i,l}^{\tp} \widehat{\bP}^{\perp}_{\mathcal{A}^*_U} 
    \vecz(\bE)}{mn \lambda_U}
$$
for a ($1/2$)-covering $\mathcal{N}_i = \{\bu_j\}$ of size at most $5^{m_i} \leq 5^M$.
By assumption, $\vecop{\bE}$ has independent $\sigma$-subgaussian coordinates, so each variable for fixed $i, l, j$ is subgaussian with parameter at most
$$
    \frac{2\sigma}{mn \lambda_U} \sqrt{\bu_j^{\tp}\hbX_{i,l}^{\tp} \widehat{\bP}^{\perp}_{\mathcal{A}^*_U} \hbX_{i,l} \bu_j } \leq \frac{2 \sigma \norm{ \widetilde{\bV}}_{1,2}}{mn \lambda_U}.
$$
By a union bound, for any $t_{\alpha} > 0$,
$$
\mathbb{P}\left( \max_{i,l} \left\lVert \frac{\hbX_{i,l}^{\tp} 
\vecz(\bE)}{mn \lambda_U} \right\rVert_2 > t_{\alpha} \right) \leq 2 pd 5^M \exp\left\{ - \left( \frac{c_E mn \lambda_U t_{\alpha}}{2 \sigma \norm{ \widetilde{\bV}}_{1,2}} \right)^2 \right\}
$$
for a universal constant $c_E > 0$. Setting $t_{\alpha} = (1-2\alpha)/3(1-\alpha)$, a sufficient condition to make the (decreasing in $\lambda_U$) RHS $\leq \eta$ is
\begin{equation} \label{score_dom}
    \lambda_U \geq \frac{6 (1 - \alpha) \sigma \norm{ \widetilde{\bV}}_{1,2}}{c_E mn (1 - 2\alpha) \omega_{\min}} \left\{ \sqrt{M\log(5)} + \sqrt{\log (2pd/\eta)} \right\}
\end{equation}
Since \eqref{score_dom} implies
\begin{align*}
\lambda_U &\geq \frac{6 (1 - \alpha) \sigma \norm{ \widetilde{\bV}}_{1,2}}{c_E mn (1 - 2\alpha) \omega_{\min}} \left\{ M\log(5) + \log (2pd/\eta) \right\}^{1/2} \\
     \implies \quad \left( \frac{c_E mn \omega_{\min} t_{\alpha}}{2 \sigma \norm{ \widetilde{\bV}}_{1,2}} \right)^2 \lambda_U^2 &\geq M \log(5) + \log(2pd/\eta) \\
    \implies \quad \eta &\geq 2 pd 5^M \exp\left\{ - \left( \frac{c_E mn \lambda_U t_{\alpha}}{2 \sigma \norm{ \widetilde{\bV}}_{1,2}} \right)^2 \right\}.
\end{align*}
using the elementary inequality $\sqrt{x} + \sqrt{y} \geq \sqrt{x + y}$.

For analysis of Terms (III) and (IV), again note that by assumption, $\XAi - \XAistar$ and $\bbeta^*_{\mathcal{A}^*_i}$ must be submatrices of $\bX_U(\widetilde{\bV}\bP^{\tp}_{r,d}) - \bX_U(\bV_*)$ and $\vecz(\bU_*)$ respectively, since the additional spurious columns in $\widetilde{\bV}$ correspond to inactive groups.

\paragraph{Term (III)}
\begin{align*}
 \| \hbX_{i,l}^{\tp} (\XAstar - \XA) \bbeta^*_{\mathcal{A}^*_U} \|_2 &= \| \hbX_{i,l}^{\tp} (\XAi - \XAistar) \bbeta^*_{\mathcal{A}^*_i} \|_2 \\
 &\leq \sum_{l'=1}^r \| \hbX_{i,l}^{\tp} (\XAi - \XAistar) \bP_{i,l'} \bP_{i,l'} \bbeta^*_{\mathcal{A}^*_i} \|_2 \\
 &\leq \sum_{l'=1}^r \| \hbX_{i,l}^{\tp} (\hbX_{i,l'} - \bX^*_{i,l'}) \|_2 \| \bu_{\mathcal{G}_i,l'} \|_2 \\
&= \sum_{l'=1}^r | \hat{\bv}_l^{\tp}(\hat{\bv}_{l'} - \bv^*_{l'}) | ~ \|\bu_{\mathcal{G}_i,l'} \|_2 \\
 &\leq r C_{\beta_U} \lVert \widetilde{\bV} \rVert_{1,2} \epsilon_V 
\end{align*}
where $C_{\beta_U} = \max \|\bu_{\mathcal{G}_i,l'} \|_2$, and $\epsilon_V = \| \bV_* - \widetilde{\bV} \bP_{r,d}^{\tp} \|_{1,2}$. 

\paragraph{Term (IV)} Denote
$$
    \bh = \XA^{\tp} (\XAstar - \XA) \bbeta^*_{\mathcal{A}^*_U}
$$
Then similar to Term (I),
\begin{align*}
    &\| \hbX_{i,l}^{\tp} \XA ( \XA^{\tp} \XA )^{-1} \XA^{\tp} (\XAstar - \XA) \bbeta^*_{\mathcal{A}^*_U} \|_2 \\
    = &\| \hbX_{i,l}^{\tp} \XA ( \XA^{\tp} \XA )^{-1} \bh \|_2 \\
    \leq &\frac{\alpha}{1-\alpha} \max_{l' \neq l} \| \bP_{i,l'} \bh \|_2 \\
    = &\frac{\alpha}{1-\alpha} \max_{l' \neq l} \| \hbX_{i,l'}^{\tp} (\XAstar - \XA) \bbeta^*_{\mathcal{A}^*_U} \|_2 \\
    \leq& \frac{\alpha r C_{\beta_U} \lVert \widetilde{\bV} \rVert_{1,2} \epsilon_V}{1 - \alpha}
\end{align*}
where the last inequality uses the same derivation as Term (III). Thus we want
$$
\frac{r C_{\beta_U} \lVert \widetilde{\bV} \rVert_{1,2} \epsilon_V}{(1 - \alpha) nm \lambda_U} \leq \omega_{\min} \frac{1 - 2\alpha}{3(1-\alpha)},
$$
or equivalently
$$
    \lambda_U \geq \frac{3 r C_{\beta_U} \lVert \widetilde{\bV} \rVert_{1,2} \epsilon_V}{mn(1 - 2\alpha)\omega_{\min}}.
$$
On the event where the bound for Term (II) holds uniformly over all groups, combining all four terms gives $\max_{(i,l)\in \mathcal{A}_U^{*c}}\|\hat{\bz}_{i,l}\|_2 < \omega_{\min}$, provided that
\begin{align*}
    \lambda_U > \frac{3 \lVert \widetilde{\bV} \rVert_{1,2} }{mn(1-2\alpha)\omega_{\min}} \max\bigg\{ \frac{2 (1 - \alpha) \sigma}{c_E} \left\{ \sqrt{M\log(5)} + \sqrt{\log (2pd) + \log(1/\eta)} \right\}, r C_{\beta_U} \epsilon_V \bigg\}.
\end{align*}
Since this event holds with probability at least $1-\eta$, strict dual feasibility holds with probability at least $1-\eta$. Therefore $\widehat{\mathcal{A}}_U\subseteq \mathcal{A}_U^*$.

To prove exact recovery, start from \eqref{eq:optimality_beta}, and project onto the coordinates of some active group $(i,l) \in \mathcal{A}^*_U$.
Using the supporting lemma with $\gamma^{**} = 1 / \kappa \lVert \widetilde{\bV} \rVert^2_{1,2}$, we have
\begin{align*}
    \lVert \bP_{i,l} \left( \widehat{\bbeta}_{\mathcal{A}^*_U} - \bbeta^*_{\mathcal{A}^*_U}\right) \rVert_2 
    &= \lVert \bP_{i,l} ( \XA^{\tp} \XA )^{-1} ( \XA^{\tp} \widehat{\bm{\xi}} - (mn) \lambda_U \hat{\bz}_{\mathcal{A}^*_U} ) \rVert_2 \\
    &\leq \frac{\gamma^{**}}{1 - \alpha} \left\{ \lVert \hbX^{\tp}_{i,l}\widehat{\bm{\xi}} \rVert_2 + mn \lambda_U \lVert \hat{\bz}_{i,l} \rVert_2 \right\} \\
    &\leq \frac{\gamma^{**}}{1 - \alpha} \left\{ \lVert \hbX^{\tp}_{i,l}\vecz(\bE) \rVert_2 + \lVert \hbX^{\tp}_{i,l} (\XAstar - \XA) \bbeta^*_{\mathcal{A}^*_U} \rVert_2 + mn \omega_{\max} \lambda_U \right\} \\
    &\leq \frac{\gamma^{**}}{1 - \alpha} \left\{ t_{\alpha} mn \omega_{\min} \lambda_U + r C_{\beta_U} \epsilon_V \lVert \widetilde{\bV} \rVert_{1,2} + mn \omega_{\max} \lambda_U \right\} \\
    &\leq \frac{\gamma^{**}}{1 - \alpha} \left\{ t_{\alpha} mn \omega_{\min} \lambda_U + \frac{1-2\alpha}{3} mn \omega_{\min} \lambda_U + mn \omega_{\max} \lambda_U \right\}
\end{align*}
with probability at least $1-\eta$. Define
$$
    C_{\alpha} = t_{\alpha} + \frac{1 - 2\alpha}{3} + \delta
$$
for a small constant $\delta > 0$, and recall that $c_{\beta_U} = \min_{(i,l)\in\mathcal{A}^*_U} \lVert \bbeta^*_{i,l} \rVert_2$ denotes the minimum active signal.
All the active groups will be selected with probability at least $1 - \eta$ if
$$
    c_{\beta_U} \geq \frac{\gamma^{**} mn \lambda_U}{1 - \alpha} ( C_{\alpha} \omega_{\min} + \omega_{\max}) \quad \iff \quad \lambda_U \leq \frac{c_{\beta_U} (1 - \alpha) \kappa \lVert \widetilde{\bV} \rVert^2_{1,2} }{mn ( C_{\alpha} \omega_{\min} + \omega_{\max})} 
$$
Combining this with the strict dual feasibility argument for inactive groups gives
$\widehat{\mathcal A}_U=\mathcal A_U^*$, as desired.

\end{proof}

\subsection{Supporting lemma}

\begin{lemma} \label{lem:block_inverse}
    Suppose $\bM = [ \bB_{kl} ]$ is a symmetric block matrix which satisfies the diagonal dominance conditions
    $$
        \lVert \bB_{kk}^{-1} \rVert_2 \leq \gamma, \quad \sum_{l \neq k} \lVert \bB_{kk}^{-1} \bB_{kl} \rVert_2 \leq \alpha < 1
    $$
    for all $k$.
    Then for any vector $\bv$, and projection operator $\bP_k$ onto the group $k$ coordinates,
    $$
        \max_k \lVert \bP_k \bM^{-1} \bv \rVert_2 \leq \frac{\gamma}{1-\alpha} \max_k \lVert \bv_k \rVert_2.
    $$
\end{lemma}

\begin{proof}
    Let $\bu = \bM^{-1}\bv$. Expanding block-wise, 
    \begin{align*}
        &\bP_k \bM \bu = \bv_k \\
        \implies \quad &\bB_{kk} \bu_k = \bv_k - \sum_{l \neq k} \bB_{kl} \bu_l \\
        \implies \quad &\bu_k = \bB_{kk}^{-1} \bv_k - \sum_{l \neq k} \bB_{kk}^{-1}\bB_{kl}\bu_l \\
        \implies \quad &\lVert \bu_k \rVert_2 \leq \gamma \max_k \lVert \bv_k \rVert_2 + \alpha \max_k \lVert \bu_k \rVert_2 \\
        \implies \quad &\max_k \lVert \bu_k \rVert_2 \leq \frac{\gamma}{1 - \alpha} \max_k \lVert \bv_k \rVert_2.
    \end{align*}
\end{proof}

\end{document}